\documentclass[journal]{IEEEtran}

\usepackage{amsmath, amssymb, bbm}
\usepackage{graphicx, pifont} 
\let\oldding\ding
\renewcommand{\ding}[2][1]{\scalebox{#1}{\oldding{#2}}}

\usepackage{caption}
\usepackage{subfigure}

\newtheorem{theorem}{Theorem}
\newtheorem{proof}{Proof}

\usepackage{pgfplots,pgfplotstable,booktabs}

\pgfplotsset{compat=1.5}

\usepackage{etoolbox}

\apptocmd{\sloppy}{\hbadness 10000\relax}{}{}

\makeatletter
\newcommand*{\rom}[1]{\expandafter\@slowromancap\romannumeral #1@}
\makeatother

\newcommand*\circled[1]{\tikz[baseline=(char.base)]{
            \node[shape=circle,draw,inner sep=0.2pt] (char) {#1};}}

\usepackage{tikz}

\usepackage{cite}

\begin{document}

\title{Robust Stability Assessment in the Presence of Load Dynamics Uncertainty}

\author{Hung~D.~Nguyen,~\IEEEmembership{Student Member,~IEEE,} and~Konstantin~Turitsyn,~\IEEEmembership{Member,~IEEE}
\thanks{Hung D. Nguyen and Konstantin Turitsyn are with the Department of Mechanical Engineering, Massachusetts Institute of Technology, Cambridge, MA, 02139 USA e-mail: hunghtd@mit.edu and turitsyn@mit.edu}}

\markboth{IEEE Transactions on Power Systems 2015, in press}%
 {Shell \MakeLowercase{\textit{et al.}}: Robust Stability Assessment in the Presence of Load Dynamics Uncertainty}%

\maketitle

\begin{abstract}

Dynamic response of loads has a significant effect on system stability and directly determines the stability margin of the operating point. Inherent uncertainty and natural variability of load models make the stability assessment especially difficult and may compromise the security of the system. We propose a novel mathematical ``robust stability'' criterion for the assessment of small-signal stability of operating points. Whenever the criterion is satisfied for a given operating point, it provides mathematical guarantees that the operating point will be stable with respect to small disturbances for any dynamic response of the loads. The criterion can be naturally used for identification of operating regions secure from the occurrence of Hopf bifurcation. Several possible applications of the criterion are discussed, most importantly the concept of Robust Stability Assessment (RSA) that could be integrated in dynamic security assessment packages and used in contingency screening and other planning and operational studies.
\end{abstract}


\section{Introduction}

Loss of stability of power systems usually results in some of the most dramatic scenarios of power system failure and has played an important role in most of the recent blackout. The dynamic of response of loads affects the voltage and to lesser extend angular stability in most important way \cite{kundur2004definition}. The loads affect the overall system behavior and may lead to loss of stability because of insufficient damping \cite{Hiskens95damping}. Typically the loss of stability of the system occurs via Hopf bifurcation \cite{dobson1989towards, canizares1994transcritical, Chow2006applied}, when some part of the upper branch of the nose curve becomes unstable. The load response was shown to play a major role in this scenario for example in \cite{Overbye1994, PaiDynamics, Venkatasubramanian1992, dobson1993new}. Hereafter, whenever we mention stability, we mean small-disturbance stability that associates with a particular operating point.

Loads, by definition, represent an aggregate of hundreds or thousands of individual devices such as motors, lighting, and electrical appliances \cite{standardloadmodel}. Load modeling has been a subject of intensive research for several decades \cite{Concordia,Karrison94,Meyer1982,Mansour94,Loadmodel}; however, it is still a rather open subject. Even though some certain types of loads such as aluminum or steel plant, and pumped hydroelectric storage are considered as well-identified ones \cite{ilic2000dynamics}; due to its natural complexity and uncertainty, load dynamics, in general, may be never known completely in operational planning, operation, and control \cite{Hiskens2000bounding, Hiskens2006sens}. The lack of knowledge about the dynamic characteristic of each individual component due to poor measurements, modeling, and exchange information, as well as the uncertainties in components/customers behaviors via switching events contribute to load uncertainties. Hence, loads are the main source of uncertainty \cite{Hiskens2006sens} that undermines the accuracy of the power dynamic models used by system operators all over the world.

Incorporation of the uncertainty into existing models is essential for improving the system security usually defined as the ability of the system to withstand credible disturbances/contingencies while maintaining power delivery services continuity \cite{KundurDSA, Oren2007management}. The future power systems will likely be exposed to higher levels of overall stress and complexity due to penetration of renewable generators, and more intelligent loads, deregulation of the system, and introduction of short-time scale power markets. Secure operation of these systems will necessarily require the operator to track the voltage stability boundary with new generation of security assessment tools providing comprehensive, fast and accurate assessment \cite{EPRIVSTAB1993}. This work addressed the need in ``robust" security assessment tools that can provide security guarantees even in the presence of modeling uncertainty. 

In \cite{dobson2003distance,dobson1991direct,dobson1992iterative}, several techniques were developed that rely on transversality conditions for quantifying the distance to various types of bifurcation including Hopf bifurcation in multidimensional parameter space. These techniques ensure robust stability of the equilibrium associated with nominal parameter $\Lambda_0$. Although they could be naturally extended to a uncertainty in small subspace of parameters, there extension to situations when the space of uncertain parameters has high dimension. In this paper, we provide robust stability certificate in multidimensional space of certain system parameters. Unlike the works mentioned above we do without tracking the most dangerous direction, rather we indicate whether such directions exist or not. Hence, we do not attempt to find the unstable points associated with some certain critical parameters.

The existence of robust stability certificate and whole region of operating points that are certified to be robust stable provides new practical alternatives for dealing with load dynamics uncertainty. It has been noted in \cite{dobson2011irrelevance} that traditional ``voltage collapse'' instability is not affected by the load dynamics as it corresponds to saddle-node bifurcation, where the equilibrium point disappears altogether. At the same time for the more common Hopf bifurcation it was argued in \cite{hiskens2006significance} that sensitivity analysis of the system trajectories may provide enough information to assess the risks associated with common disturbances. Moreover, whenever the system operates in the robust stability regime, the stability can be certified even without knowing the dynamic characteristics of the load altogether. The stability of the system can be certified simply by analyzing the static characteristics of the loads in combination with well-understood dynamic models of generators. In this sense, we argue that accurate modeling of the loads is essential only when the system operates in the intermediate regimes of the nose curves or the PV curves, between the robust stability region and the saddle-node bifurcation on the nose tip. 

The structure of the paper and the main contributions are summarized below. After introducing our modeling assumptions in \ref{sec:LoadModeling} we derive the novel robust stability criterion in section \ref{sec:robuststability}. Then, we propose a practical algorithm RSA for robust stability certification. In section \ref{sec:simulation} we perform various simulations with several test cases from $2$-bus system to WSCC $3$-machine, $9$-bus system and the IEEE $39$-bus New England system to illustrate the concept of robust stability and RSA. The dynamic simulations are implemented in SystemModeler $4.0$ and the computations are performed in Mathematica $10$ and with the help of CVX program, a package for convex programming. Then in section \ref{sec:applications} we discuss the proposed applications of the algorithm, and possible extensions to other kinds of uncertainty. Finally, the non-certified robust stability region is discussed in section \ref{sec:nonrobust}.

\section{Voltage stability and load dynamics} \label{sec:loadmodel}

\subsection{Voltage stability}
While the power system operates in stressed heavily loaded regime it may be prone to subject to voltage stability problems. The secure operating region is confined by voltage stability boundary. As a common practice, static voltage stability criteria is widely used by system operators \cite{xie2007novel,LeXieQSVS}. Moreover, it has been argued that static analysis is preferred over dynamic approach \cite{morison1993voltage}. At the same time it has been reported in many works that Hopf bifurcation may destabilize the system before it reaches the static stability limits \cite{canizares1994transcritical}.

\begin{figure}[ht]
    \centering
    \includegraphics[width=5.5cm]{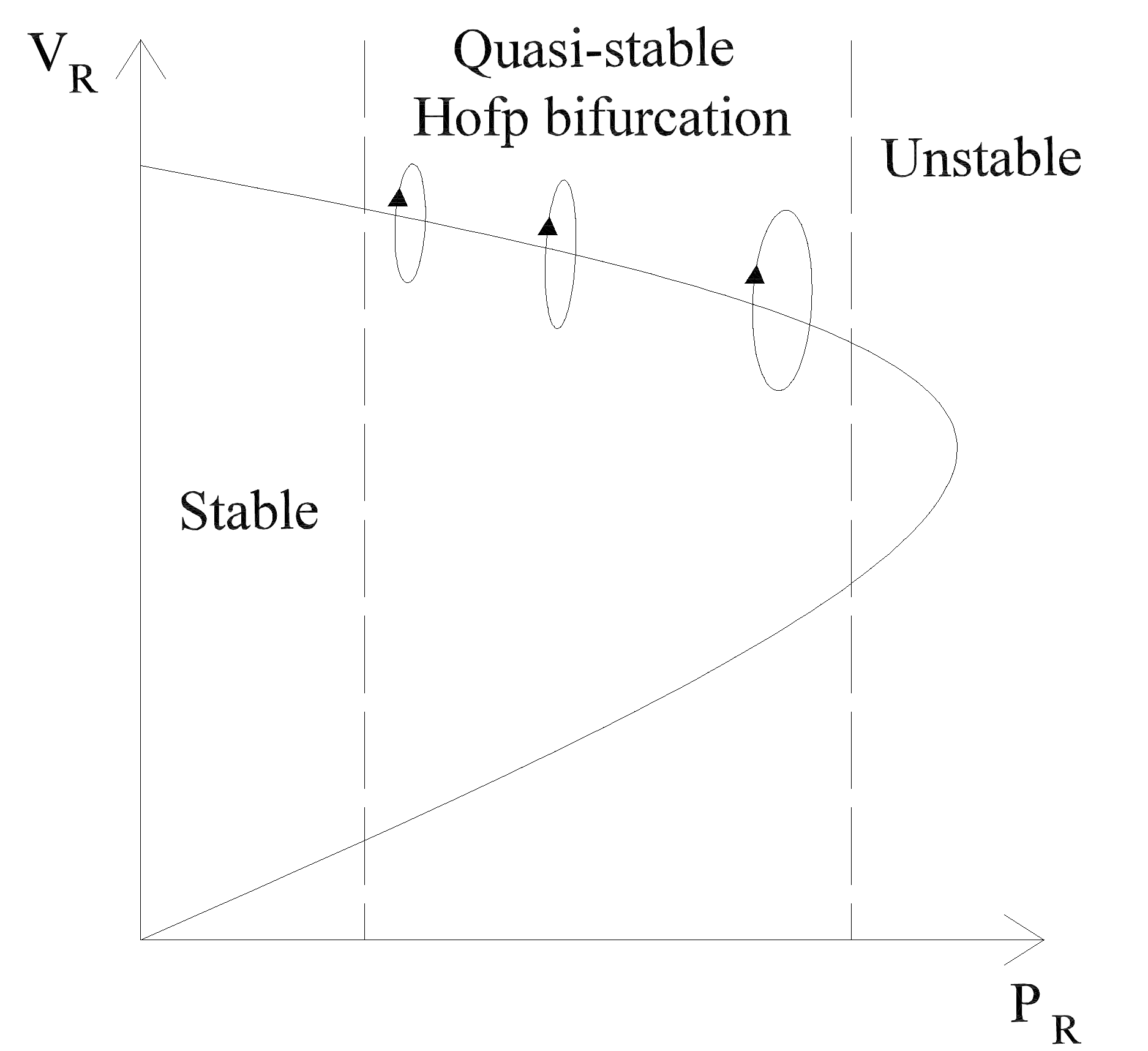}
    \caption{Qualitative visualization of Hopf bifurcation \cite{knightemergency}}
    \label{fig:qualHopf}
\end{figure}

Under some particular conditions, Hopf bifurcation may not occur \cite{Chang1988stability} but typically, Hopf bifurcation determines the stability margins of most common systems \cite{canizares1994transcritical} when the system exhibits Hopf bifurcation before it reaches the saddle-node bifurcation point or the tip of the nose curve. This situation can happen in the quasi-stable Hopf bifurcation region shown in Figure \ref{fig:qualHopf}. The term quasi-stability used in power engineering is related to the oscillatory behavior of the system that is observed after the occurrence Hopf Bifurcation \cite{knightemergency}. Detecting the loadability limits associated with the bifurcation is a much more complicated problem in comparison to the static stability analysis associated with the saddle-node bifurcation \cite{AlvaradoHopfissue, DobsonSenHopf, Ajjarapu2005fastHopf}. Some realistic examples of finding Hopf bifurcation point can be found in \cite{ilic2000dynamics} and related works. In this context, the key contribution of this work is an alternative approach based on robust stability. Whenever the robust criterion criterion is satisfied, the system is mathematically guaranteed that Hopf bifurcation cannot occur.
\subsection{Dynamic load modeling} \label{sec:LoadModeling}

The stability of any operating point and the position of the Hopf bifurcation on the nose curve depends on the dynamical behavior of loads on individual buses. Traditional models of load dynamics are based on combination of differential and algebraic equations for the load state. In steady state the loads can be characterized by their static characteristics $P^s(V,\omega)$ and $Q^s(V,\omega)$ which describe the dependence of the active and reactive power consumption levels $P,Q$ on the load bus voltage level $V$ and system frequency $\omega$. The dynamic state of the loads is typically characterized by single state variable $x$ that represents the internal state of the system, for example the average slip of the induction motors. Whenever the composition of the loads on a single bus is highly heterogeneous, it may be represented by a parallel interconnection of several components characterized by different models. At any moment of time the load consuming active power $P$ and reactive power $Q$ can be characterized by the effective conductance $g = P/V^2$ and susceptance $b = Q/V^2$. The first order dynamic model for the conductance representing the dynamics of the internal state of the load can be than written in a general form as:
\begin{align}\label{eq:abstract}
& \dot g = F(g, V, \omega)
\end{align}
The right hand side of this equation is not arbitrary and should have the equilibrium point corresponding to the steady state characteristic of the load. Hence whenever the active power consumption is equal to steady rate, so $P=g V^2 =  P^s(V,\omega)$, the right hand side of \eqref{eq:abstract} should vanish, so $F(P^s(V,\omega)/V^2,V,\omega) = 0$. Any function $F$ that satisfies this relation can be rewritten as $F = \tau^{-1}(P^s(V,\omega) - g V^2)$. In this form, the factor $\tau$ generally depends on voltage and frequency and can be interpreted as instant relaxation rate of the load. Whenever the load is stable when connected to an infinite slack bus, the factor $\tau$ can be trivially shown to be positive, so $\tau > 0$. The same mathematical form and analysis also apply to the load susceptance.

This discussion allows us to conclude that for the purposes of small-signal stability studies the first order models of the loads can be represented as 
\begin{align}\label{eq:geqsimple}
\tau_{gk}  \dot{g_k} = - (g_k V_k^2 - P^s_k), \\\label{eq:beqsimple}
\tau_{bk}  \dot{b_k} = - (b_k V_k^2 - Q^s_k).
\end{align}
Here the index $k$ runs over all load buses in the system, the factors $\tau_{gk}$, $\tau_{bk}$ represent the uncertainties in the dynamic models, that can be also interpreted as relaxation time. The factors $P^s_k$ and $Q^s_k$ represent the voltage dependent static characteristic of the loads. 

This type of load model is also introduced in \cite{taylor1994stability, Cutsem, Overbye1994}, typically for thermostatic loads. However as we have argued in \cite{nguyen2014voltage} this model can naturally be used to represent the standard models for thermostatically controlled loads, induction motors, power electronic converters, aggregate effects of otherwise unmodelled distribution Load Tap Changer (LTC) transformers etc. The static loads can be also naturally modeled within this framework by taking the limit $\tau_{gk} \to 0$. Obviously, the range of time constants is wide, ranging is from cycles to minutes and can introduce a lot of uncertainty in the modeling process.

We finish this section by comparing the model to the two other classical load models.  Equations \eqref{eq:geqsimple} are just another form of the traditional dynamic load models introduced originally in \cite{Karrison94, Hill93}:
\begin{equation} \label{eq:generalform}
\centering
\dot{P_d}+f(P_d,V)=g(P_d,V)\,\dot{V}
\end{equation}
Here $P_d$ is the instantaneous power, that is denoted by $p_k = g_k V_k^2$ in our notations and $V$ is the bus voltage magnitude, referred to as $V_k$ in equations \eqref{eq:geqsimple}. The more specific form of these equations, known as exponential recovery model was introduced in \cite{Karrison94, Hill93}:
\begin{equation} \label{eq:exploadHill}
\centering
T_p\,\dot{P_d}+P_d=P_s(V)+k_p(V)\,\dot{V}
\end{equation}
We can recover the model (\ref{eq:generalform}) from equation (\ref{eq:geqsimple}) by taking the derivative of $g_k |V_k|^2$. This results in the following expression:
\begin{equation}  \label{peq}
 \dot{p_k} + \frac{p_k - P^s_k(V_k)}{\tau_{g\,k}}V_k^2 = 2 \frac{p_k}{V_k}\frac{d }{d t}V_k
\end{equation}

Another equivalent model was introduced in \cite{Mansour94} and \cite{Lesieutre1995}:
\begin{align}
T_p\frac{dx}{dt}=P_s(V)-P;\, P=x\,P_t(V)\\
T_q\frac{dy}{dt}=Q_s(V)-Q;\, Q=y\,Q_t(V)
\end{align}
where $x$ is the state; subscript $s$ and $t$ indicate steady state and transient values, respectively; $P_t(V)=V^\alpha$, $P_s(V)=P_0\,V^a$; $Q_t(V)=V^\beta$, $Q_s(V)=Q_0\,V^b$. This model is equivalent to \eqref{eq:geqsimple}, \eqref{eq:beqsimple} with $x=g_k$ and $y=b_k$ when $\alpha=\beta=2$.

The proposed load model can naturally represent the most common types of loads, such as induction motors, thermostatically controlled loads. Hence, we believe that the form of the load model is rather general and can be used in a variety of practically relevant problems.

For example, below we show, how the induction motor model can be embedded in our generic modeling framework. 

\begin{figure}[ht]
    \centering
    \includegraphics[width=1 \columnwidth]{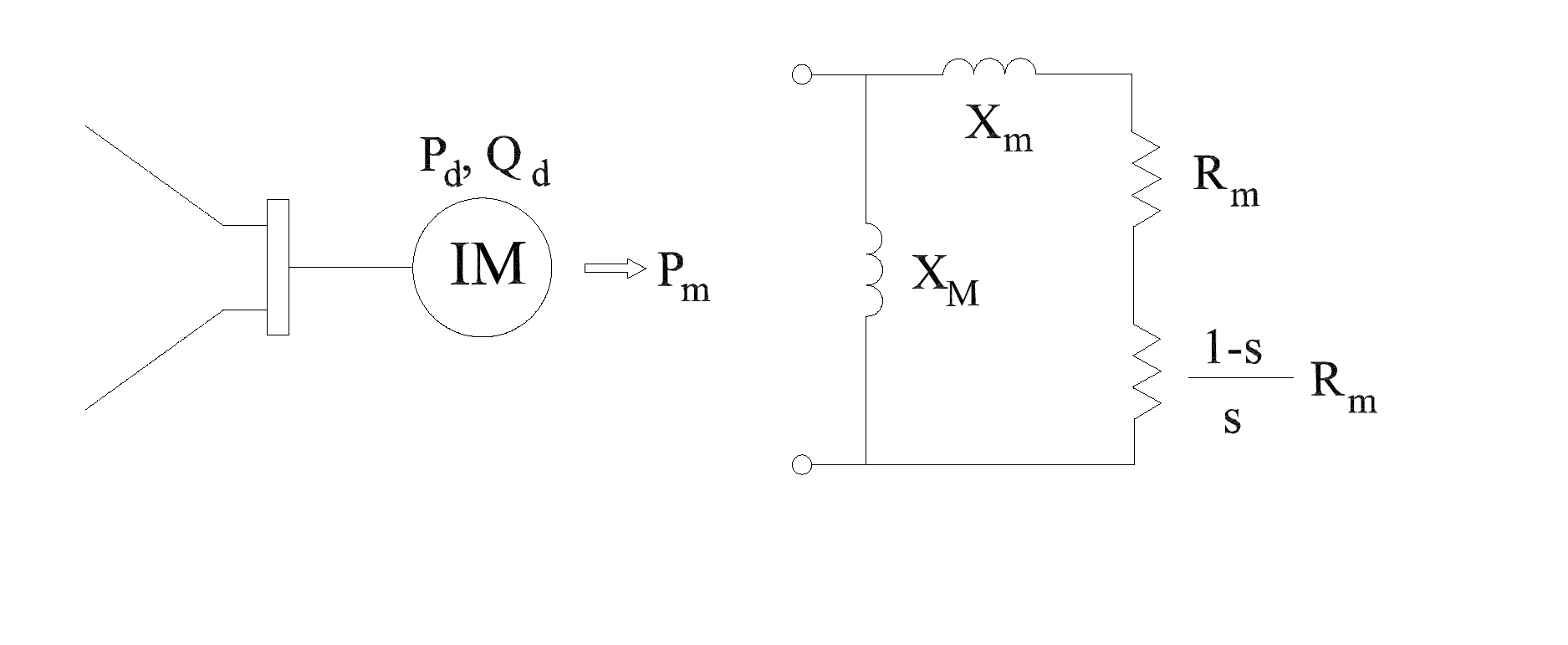}
	\caption{Induction motor load model \cite{Hill93}}
    \label{fig:motor}
\end{figure}
The induction motor depicted in Figure \ref{fig:motor} can be described as \cite{Hill93}:

\begin{equation} \label{eq:motorslipdynamics}
\centering
\dot{s}=\frac{1}{I \omega_0^2}(\frac{P_m}{1-s}-P_d)
\end{equation}
where $s$ is the motor slip, $\omega_0$ is the base frequency, $I$ is the rotor moment of inertia, $P_m$ is the mechanical power, and $P_d$ is the electric power given by
\begin{equation} \label{eq:motoractivepower}
\centering
P_d=\frac{V^2\,R_m\,s}{R_m^2+X_m^2}=V^2\,h(s)
\end{equation}
Since $P_d=h(s)\,V^2$, from (\ref{eq:motoractivepower}), we can represent the motor as the dynamic inductance with 
\begin{equation} \label{eq:motorys}
g = h(s) 
\end{equation}

In normal operating regime, this relation can be also reversed so that $s=h^{-1}(g)$.

Differentiation of the two sides of (\ref{eq:motorys}) with respect to time yields the following expression:
\begin{equation}\label{eq:motordotg}
\centering
\dot{g}=\alpha\,\frac{dh}{ds}(\frac{P_m}{1-s}-g\,V^2)
\end{equation}
where $\alpha=\frac{1}{Iw_0^2}$.
As long as $s$ can be expressed in terms of $g$ we reproduce the general form (\ref{eq:abstract}). Similar approach can be applied to most of the other types of loads, like thermostatically controlled loads, static loads behind Under-Load Tap Changers (ULTCs), and certainly the static loads which are described in more detail in Appendix \ref{app:LM}.

From (\ref{eq:motorys}) and (\ref{eq:motordotg}), the induction motor load can be modeled in the form of (\ref{eq:abstract}). More importantly, the proposed dynamic load model not only is convenient for static analysis even in non-conventional power flow regime \cite{nguyen2014voltage} but also satisfies all fundamental requirements for load models in voltage stability studies which are mentioned in \cite{Morisonload2003}.

\section{Stability theory} \label{sec:robuststability}
In this section we address the question of the small-signal stability of an operating point by first reviewing the classical stability criteria applied to the problem of voltage stability of modern power system models in subsection \ref{sec:linear} and then introduce the central result of the paper: robust stability criterion in \ref{sec:robust}.

\subsection{Linear stability} \label{sec:linear}
Most common models of power system dynamics describe the evolution of the power system in terms of a system of nonlinear differential algebraic equations of the form
\begin{align} \label{eq:DAEs}
\dot{x}&=F(x,y)\\
0 &= G (x,y)
\end{align}
where $x \in \mathbb{R}^n$, $y\in \mathbb{R}^m$ are vectors representing the state variables, algebraic variables. The state variables can be naturally decomposed in generator $x_\mathcal{G} \in \mathbb{R}^{n_\mathcal{G}}$ and load states $x_\mathcal{L}\in \mathbb{R}^{n_\mathcal{L}}$. Here $n_\mathcal{L}$ and $n_\mathcal{G}$ are the total number of states associated with loads and generators, respectively. Moreover, we assume that the subset of algebraic variables $y$ represents the bus voltages, including the voltages on load buses.

Under the assumptions above it is possible to represent (\ref{eq:DAEs}) in terms of $x_\mathcal{G}$ and $x_\mathcal{L}$ as:
\begin{align} \label{eq:geneq}
\dot{x_\mathcal{G}}&=F^\mathcal{G}(x_\mathcal{G},y)\\ \label{eq:loadeq}
\dot{x_\mathcal{L}}&=\mathcal{T}^{-1}F^\mathcal{L}(x_\mathcal{L},y)\\
0 &= G (x_\mathcal{G},x_\mathcal{L},y)
\end{align}
where $\mathcal{T}$ is a diagonal matrix with the size of $n_\mathcal{L}\times n_\mathcal{L}$ whose diagonal entries are the time constants of corresponding loads as introduced in \eqref{eq:geqsimple}; $F^\mathcal{G}$ and $F^\mathcal{L}$ are the functions associate with the sets of generators and the loads, respectively. Note, that in this representation the functions $F^\mathcal{G},F^\mathcal{L}$ and $G$ can be assumed to be known and all the uncertainty is aggregated in the matrix $\mathcal{T}$. This assumption is reasonable in the situations when the network characteristics are known, generator models are verified and static load characteristics are understood better than their dynamic response which is the case in practical situations. Note, also, that in the equations \eqref{eq:geneq} and \eqref{eq:loadeq} there is no direct coupling between the dynamics of generators and loads, as the individual load components interact only indirectly via algebraic bus voltage variables.

Small signal stability can be characterized by considering the linearized version of the equations for the deviations of state and algebraic variables from their equilibrium values. 
\begin{equation} \label{eq:linearDAE}
\centering
\begin{bmatrix}
\dot{\delta{x_\mathcal{G}}}\\\dot{\delta{x_\mathcal{L}}}\\0
\end{bmatrix} = \begin{bmatrix}
F^\mathcal{G}_{x_\mathcal{G}} & F^\mathcal{G}_{x_\mathcal{L}} & F^\mathcal{G}_{y}\\
\mathcal{T}^{-1}F^\mathcal{L}_{x_\mathcal{G}} & \mathcal{T}^{-1}F^\mathcal{L}_{x_\mathcal{L}} & \mathcal{T}^{-1}F^\mathcal{L}_{y}\\
G_{x_\mathcal{G}} & G_{x_\mathcal{L}} & G_{y}\\
\end{bmatrix} \begin{bmatrix}
\delta{x_\mathcal{G}} \\ \delta{x_\mathcal{L}} \\ \delta{y}
\end{bmatrix}
\end{equation}
where the subscripts of $x_\mathcal{G}$, $x_\mathcal{L}$, and $y$ indicate the partial derivatives with respect to the corresponding states and variables.
Away from saddle-node bifurcation the algebraic variables $\delta{y}$ can be eliminated from \eqref{eq:linearDAE} yielding
\begin{align} 
\label{eq:unreducedlinearDAE} \nonumber
\begin{bmatrix}
    \dot{\delta{x_\mathcal{G}}}\\
    \dot{\delta{x_\mathcal{L}}}
\end{bmatrix} &= A 
\begin{bmatrix}
    \delta{x_\mathcal{G}} \\ 
    \delta{x_\mathcal{L}}
\end{bmatrix} = 
\end{align}
\begin{align}
\begin{bmatrix}
F^\mathcal{G}_{x_\mathcal{G}}-F^\mathcal{G}_y G_y^{-1} G_{x_\mathcal{G}} & - F^\mathcal{G}_y G_y^{-1} G_{x_\mathcal{L}}\\
-\mathcal{T}^{-1}F^\mathcal{L}_y G_y^{-1} G_{x_\mathcal{G}} & \mathcal{T}^{-1}(F^\mathcal{L}_{x_\mathcal{L}} - F^\mathcal{L}_y G_y^{-1} G_{x_\mathcal{L}})
\end{bmatrix} 
\begin{bmatrix}
    \delta{x_\mathcal{G}} \\ 
    \delta{x_\mathcal{L}}
\end{bmatrix} \nonumber
\end{align}

This expression can be more conveniently decomposed as $A = \Lambda J$ in the following form
\begin{align}
A = A_\mathcal{T} \triangleq  \begin{bmatrix}
    \mathbbmtt{1} & 0 \\
    0 & \mathcal{T}^{-1}
\end{bmatrix}
\begin{bmatrix}
J_{\mathcal{G}\mathcal{G}} & J_{\mathcal{G}\mathcal{L}} \\
J_{\mathcal{L}\mathcal{G}} & J_{\mathcal{L}\mathcal{L}}\end{bmatrix}.
\end{align}
where $\mathbbmtt{1}$ is the identity matrix of size $n_\mathcal{G}\times n_\mathcal{G}$.

The key advantage of this decomposition is the separation of the matrix $A$ in an uncertain diagonal matrix $\mathcal{T}$ and the Jacobian matrix $J$ that does not depend on the uncertain load time constants, and depends only on the properties of the steady state equilibrium point defined in load and generator variables. 

Notably, for load models considered in this work the second row depends only on the steady-state behavior of the load, i.e. it can be computed given the load levels and voltage/frequency dependence of the steady-state active and reactive power consumption. 

According to the Lyapunov direct method, the system described by $\dot{x}=Ax$ is stable if and only if there exist a symmetric positive definite matrix $Q=Q^\top \succ 0$ such that
\begin{equation}
\centering
QA+A^\top Q \prec 0
\end{equation}
where superscript $\top$ is used for transpose operator.
However, existence of a $Q$ matrix for a given $A$ merely implies the system stability for some specific load dynamics. In the next section, we introduce the concept of robust stability that guarantees the stability of the system stability for any load time constant uncertainty, i.e. any positive definite diagonal matrix $\Lambda$. 

\subsection{Robust stability} \label{sec:robust}
As discussed previously, in this work, we assume that the operator has reliable information about the generator models and settings, and the corresponding Jacobian matrix row $J_\mathcal{G}$ is available for analysis. At the same time, we assume that the grid model and all the algebraic equations characterized by $G$ are known with high accuracy. For the load model we assume that the matrices $F^\mathcal{L}_{x_\mathcal{L}}$ and $F^\mathcal{L}_y$ describing the static characteristics of loads are known with high accuracy, however the matrix $\mathcal{T}$ representing the dynamic response is not. The goal of robust stability certificate is to guarantee that the operating point is stable for any positive definite $\mathcal{T} \succ 0$.

It is important to distinguish between two categories of load uncertainties, i.e. load level uncertainty and load dynamic uncertainty. The former relates to load level fluctuations due to various factors such as individual consumer behavior or variations in the production output of DGs. This type of uncertainty is considered in \cite{Vittaltrajsens2012, hou2012trajectory, Hiskens2006sens, overbye1991voltage, Chakrabortty2011}. On the other hand, load dynamic uncertainty concerns the unpredictability of the dynamic response of the load to small fluctuation in voltage and frequency. In this work, we only focus on the latter type of uncertainty and do not discuss the uncertainty in load variations assuming that the operating point is known. However, the regions of robust stability can be also used to account for uncertainty in load consumption levels. 

There are many sources of uncertainty in load dynamics. Apart from the natural uncertainty related to composition of power consumption devices, the level of uncertainty may increase dramatically in coming years when more small scale generators, i.e. DGs, are integrated to the systems, especially on the distribution level. When the penetration level becomes very high the traditional static voltage stability may be insufficient to assess the system security \cite{nguyen2014appearance, nguyen2014voltage}. On the other hand, the approach proposed in this work is valid, at least for non-synchronous DGs that can be modelled as a negative loads with dynamics in the form of \eqref{eq:geqsimple} and \eqref{eq:beqsimple}.

The robust stability criterion developed in the manuscript is directly linked to the concept of D-stability \cite{johnson1974sufficient,kaszkurewicz2000matrix} that are extended to model the uncertainty in a subset of state variables.  

In the following theorems we denote the set of positive definite matrices of size $n\times n$ as $\mathbb{P}_n$ and set of diagonal positive definite matrices of size $n \times n$ as $\mathbb{D}_n$. The following theorem is central to the robust stability certification of power systems.

\begin{theorem}
Assume that there exists block-diagonal positive definite block diagonal matrix $Q$, such that
\begin{equation}\label{eq:structure}
 Q = \begin{bmatrix}
Q_\mathcal{G} & 0 \\
0 & Q_\mathcal{L}\end{bmatrix},
\end{equation}
with positive definite matrix $Q_\mathcal{G} \in \mathbb{P}_{n_\mathcal{G}} $ and diagonal positive definite matrix $Q_\mathcal{L} \in \mathbb{D}_{n_\mathcal{L}} $ that satisfies  
\begin{equation}\label{eq:robust}
Q A_\mathcal{T} + A_\mathcal{T}^\top Q \prec 0
\end{equation}
for some $\mathcal{T} > 0 $. In this case the system is robust stable, i.e. in other words, for any diagonal $\tilde{\mathcal{T}} \in \mathbb{D}_{n_\mathcal{L}}$ there exists $\tilde Q \succ 0$ such that $\tilde{Q} A_{\tilde{\mathcal{T}}} + A^\top_{\tilde{\mathcal{T}}} \tilde{Q} \prec 0$
\end{theorem}
\begin{proof}
Consider the matrix $\tilde Q = \tilde Q^\top = Q \mathcal{T} \tilde{\mathcal{T}}^{-1}$. Due to block-diagonal structure of $Q$ we have $\tilde Q A_{\tilde{\mathcal{T}}} = Q A_\mathcal{T}$ and at the same time $A_{\tilde{\mathcal{T}}}^\top \tilde Q = A_\mathcal{T}^\top Q $, so $\tilde Q A_{\tilde{\mathcal{T}}} + A_{\tilde{\mathcal{T}}}^\top \tilde Q  = Q A_\mathcal{T} + A_\mathcal{T}^\top Q \prec 0$.
\end{proof}

Note, that the condition \eqref{eq:robust} first reported in the framework of D-stability \cite{johnson1974sufficient, kaszkurewicz2000matrix} only establishes a sufficient criterion for robust stability. To our knowledge no computationally tractable necessary and sufficient criteria reported for D-stability have been reported in the literature. The only exception is the set of results on the so-called positive matrices \cite{knorn2009linear} for which the existence of diagonal Lyapunov function is a necessary condition for stability. Positive matrices are characterized by negative off-diagonal components. The question of whether they can be used to describe power system dynamics is interesting and worth exploring, but is outside of the scope of this study. 

The problem of checking whether the block diagonal matrix $Q$ exists for given $A_\mathcal{G}$, $A_\mathcal{L}$  and $\mathcal{T}$ is easy and can be accomplished by solving the following semi-definite programming (SDP) problem.
\begin{align}\label{eq:SDP}
& \max_{Q} \rho \\ \nonumber
\textrm{subject to:  } & Q A_\mathcal{T} + A_\mathcal{T}^\top Q + \rho \mathbbmtt{1} \prec 0 \\
& Q \succ 0 \nonumber \\
& \textrm{tr}(Q) = 1. \nonumber
\end{align}
Here the optimization is carried over the matrices $Q$ with structure defined in \eqref{eq:structure}. The condition $\textrm{tr}(Q) $ fixes the overall normalization of the Lyapunov function. Whenever the resulting value $\rho$ is positive the system is guaranteed to be robust stable. The complexity of this procedure is polynomial in the size of the system. In recent years mathematically similar procedures have been successfully applied in the context of optimal power flow approaches \cite{jabr2006radial,lavaei2010convexification}, and more recently for power system security assessment purposes \cite{molzahn2013sufficient}. It has been shown in a number of papers, that even large scale systems admit fast analysis with SDP algorithms \cite{lesieutre2011examining}. 

However, from \eqref{eq:loadeq}, one can see that the proposed robust stability criterion requires the equilibrium to be independent on uncertain parameters, for example the time constants of the loads. Fortunately, the standard control systems in generators and other components normally satisfy this requirement. This can be seen by looking at the equations for the system equilibrium point, like load flow equations and observe that they don't depend on the dynamic time constants of governors, AVR and loads. 

In this work we illustrate the approach by considering the load dynamic uncertainties. In real power systems, the dynamics of generators and Flexible AC Transmission Systems (FACTS) devices are also the sources of uncertainties \cite{Kundur,anderson1977power,Vittal2002,chakrabortty2008estimation}. The generators and the system uncertainties cause much difficult in designing effective Power System Stabilizer (PSS) and other controllers \cite{4075245,372584}. As mentioned before, as long as these uncertainties do not alter the system equilibrium, the proposed robust stability criterion can be applied to access the system stability. In this case, all known dynamic components can be grouped in set $\mathcal{G}$ and unknown dynamic ones can be classified in set $\mathcal{L}$.

\section{Proposed applications} \label{sec:applications}
In this section we discuss the possible applications of the mathematical techniques explained above. 

\subsection{Dynamic Security Assessment (DSA)}
DSA are used to analyze the security of power systems and assess various types of stabilities such as voltage stability in Voltage Stability Assessment (VSA) and transient stability which is assessed in Transient Stability Assessment (TSA). The configuration of the DSA integrated into the Energy Management System (EMS) is discussed in details in \cite{KundurDSA}. Depending on the purpose of the assessment and the time-scale of the function of interest, the input of DSA may be different. Typical DSA assess the stability of a given operating state determined either from Supervisory Control and Data Acquisition (SCADA) or Phasor Measurement Unit (PMU) measurement tools or constructed in framework of scenario analysis for planning or operation purposes. Being a fundamental component of DSA toolbox, the main goal of VSA is to certify pre- and post-contingency voltage stability and calculate the voltage stability margin. The contingency set typically includes major equipment outages such as generator, transformer, line tripping. $N-1$ security set is normally of interest \cite{savulescu2009real,KundurDSA,fouad1988dynamic}.

Brute-force accounting for load dynamics and other uncertainties in traditional VSA is computationally expensive due to large number of scenarios that need to be analyzed. An alternative proposed here and discussed in more details in section \ref{sec:RSA} is based on the worst case scenario analysis and relies only on the analysis of static characteristics of the loads and well-understood dynamic characteristics of the generators. Hence it eliminates the need for computationally expensive dynamic simulations and stochastic Monte Carlo approached to modeling the uncertainty. 

Typically, the objective of the DSA module is to assess the system stability margins and its behavior in major contingencies. At the input, the DSA module admits a scenario which includes: i) a power flow base case which describes a snapshot of the system conditions; ii) dynamic data of the system; iii) set of critical disturbances. The output from the DSA module is composed of the system stability and corresponding margins. The work \cite{huang2002intelligent} describes DSA in more details from the perspectives of both traditional approaches in off-line analysis as well as intelligent system (IS) based one for on-line assessments.

It is worth to distinguish the two main classes of security assessment, i.e. Static Security Assessment (SSA) and DSA. SSA concerns whether the operating constraints are satisfied, i.e. whether the post-contingency voltage lies within the acceptable range, whereas DSA looks for the system stability. In some cases, acceptable voltage levels may imply that the system is stable. However, in general, this relationship is not such simple. Therefore, the system stability needs to be assessed thoroughly in the framework of DSA.

\subsection{Robust Stability Assessment}\label{sec:RSA}

The algorithms developed in this work can form the foundation of a potentially more powerful technique that we call Robust Stability Assessment (RSA). Specifically we propose to use RSA to develop the fast screening phase of VSA in an online DSA that is required to be fast enough to either automatically or manually choose the proper remedial control actions. For an effective and powerful VSA, the accuracy and the speed of computation are the two most crucial and challenging issues. As previous mentioned, the accuracy of VSA is affected due to uncertainties. RSA is able to eliminate such errors. Moreover, the fast algorithm of RSA is extremely helpful to speed up the program, especially when it relies on deterministic method that exhaustively screens contingency and searches for secure limits. Even for intelligent system based VSA, RSA is still able to help to remove a significant number of possibilities. The efficiently computational aspect of the proposed algorithm can be easily scale to bulk systems which is impossible for traditional dynamic approaches while rendering the meaning of dynamic stability assessment.

Within this approach in RSA, the stability is certified not for a single mathematical model of a system, but rather for the whole set of systems defined by different realizations of uncertain elements. The key steps required for performing the Robust Stability Assessment are explained below:

\begin{enumerate}
\item \textbf{Input} 
The input of RSA is an equilibrium configuration of the system characterized by the levels of load consumption, network model, and dynamic model of the generators.
\item \textbf{Initialization} On the initialization stage the algorithm defines the model of the system by introducing the uncertain model of the load. In the simplest approach the load buses are modeled as time dependent impedances as discussed in section \ref{sec:LoadModeling}. In the framework of more advanced approaches it may be reasonable to separate the actual loads into static components, well-defined dynamic ones (like aluminum smelters) and finally the uncertain dynamic loads. Only the uncertain components should be incorporated in the $x_\mathcal{L}$ part of the dynamical system descriptions, whereas all the other components should be modeled as known ones and described by the vector $x_\mathcal{G}$.

\item \textbf{Linearization} The dynamic model of the system is linearized and the matrix $A_\mathcal{T}$ is calculated for some arbitrarily chosen load relaxation time constants matrix $\mathcal{T}$. As explained in previous section the choice of initial $\mathcal{T}$ does not have any effect on the outcome of the analysis.
\item \textbf{Optimization} The Semi-Definite Programming problem \eqref{eq:SDP} is solved for the constructed matrix $A_\mathcal{T}$. Whenever the resulting value $\rho$ is positive the equilibrium point is certified to be robust stable, i.e. it is provably stable for any matrix $\mathcal{T}$.
\item \textbf{Direct Analysis} As the condition $\rho>0$ from \eqref{eq:SDP} is only sufficient but not necessary, whenever the result of optimization results in negative $\rho$, nothing can be said about the stability of the system. The user of RSA has to rely on other probabilistic or deterministic techniques to assess the probability of having stable system given the uncertainty in load dynamics.
\end{enumerate}
RSA can be naturally incorporated in several planning and operational studies that are described below.

\subsection{RSA for deterministic stability assessment}
One specific application of the RSA approach is the deterministic stability assessment that is regularly performed during power system operation. At any moment of time, the system operators need to know the following \cite{KundurDSA}.

\begin{enumerate}
\item Whether the current state is secure
\item Whether the system will remain secure after the next several minute changes
\item If the system is insecure, what countermeasures need to be carried out
\end{enumerate}

The general deterministic stability assessment answers these questions via the following sequence of steps \cite{mccalley2004probabilistic}:

\begin{enumerate}
\item Develop the power flow base cases for the study
\item Select the contingency set
\item Select parameters in the expected operating range \label{paramstep}
\item Identify security constraint violations
\item Find the security boundary
\item Construct the comprehensive reports like plots or tables by combining all the security boundaries
\end{enumerate}
Robust stability technique naturally fits in this process without any  adjustments to the logic. The key advantage of the RSA is its ability to certify the stability and security of the system even in the presence of dynamic uncertainty naturally expressed as parameter ranges in step \ref{paramstep}) above. The proposed robust stability criterion is compatible with both off-line and on-line security assessments in the presence of uncertainties. Moreover, it may also provide additional benefits for implementing real-time and distributed security assessment schemes which are still the main challenge to the current technologies \cite{fouad1988dynamic}. In this framework, the assessment has to be performed without access to full model of external entities, and the operator may represent the dynamic response of these entities via equivalent models with uncertain time-constants. Such a scheme is more robust to communication system malfunctions and potentially reduces the requirements to throughput and latency of sensing, communication and computation components. In some cases, large enough robust stability region can be directly applied in operation procedures and used as secure regimes that are displayed to the operators. Moreover, as mentioned before, RSA can access the system dynamic stability simply based on static analysis (power flow) and well-understood dynamic components, the dynamical secure regimes can be constructed in advance. Specific demonstration of the usage of robust stability in VSA is presented in section \ref{sec:WSCC} where we examine the $N-1$ contingency set of WSCC $3$-machine, $9$-bus system.

\subsection{Security Indicator} \label{sec:SI}
The optimization problem \eqref{eq:SDP} can be used not only to certify the stability of a given point but also to estimate the stability margin. Indeed, the value of $\rho$ is naturally interpreted as the worst case rate of decay of the Lyapunov function defined by $x^\top Q x$ and can be thus viewed as the worst case stability margin. The security indicator defined by $\rho$ can be used for risk monitoring purposes and can assist the system operators in designing the preventive control strategies. In the latter it is natural to optimize for control actions that ensure some minimal level of worst-case stability margin.

With additional research effort invested it should be possible to modify the security indicator defined by $\rho$ from \eqref{eq:SDP} in a way that it's value reflects the probabilities of system losing stability in the presence of random factors, such as renewable generators. To achieve this goal it is necessary to study the sensitivity of matrix $A$ with respect to random factors, and modify the term $\rho\mathbbmtt{1}$ in a way that certificate that bounds $\rho$ from below can be interpreted in probabilistic way, i.e. probability of system losing stability bounded from above.

\subsection{Stability constrained Planning and Optimization}

RSA and security indicator discussed in section \ref{sec:SI} can be also used for planning and dispatch purposes in the framework of stability or security constrained optimization. In this case the security indicator can be used as one of the optimization objectives or constraints. As closed form expression for $\rho$ does not exist, the corresponding optimization needs to rely on some iterative heuristics, like genetic algorithms. The algorithms may need to be complemented with direct approaches as described for example in \cite{makarov1998general, dobson2003distance, makarov2000computation, Hiskens2005limit}.

\section{Simulations} \label{sec:simulation}
In this section we report the results of application of the Robust Stability Certification to several common models of power systems. Moreover, RSA technique does not explicitly address the question of feasibility of the operating point, although it could be trivially extended with any kinds of voltage and current constraints. As these constraints depend on the operating point, and not on the dynamic equations, they can be checked separately from the small signal stability. Whenever the small-signal stability of the operating point needs to be analyzed, and RSA technique allows to assess stability even in the presence of load modeling uncertainty. As a matter of fact, in contingency analysis, it is essential to assess the system stability even when the voltage levels are unacceptable according to normal operating conditions.

\subsection{A $2$-bus system}

\begin{figure}[!ht]
    \centering
    \includegraphics[width=0.8 \columnwidth]{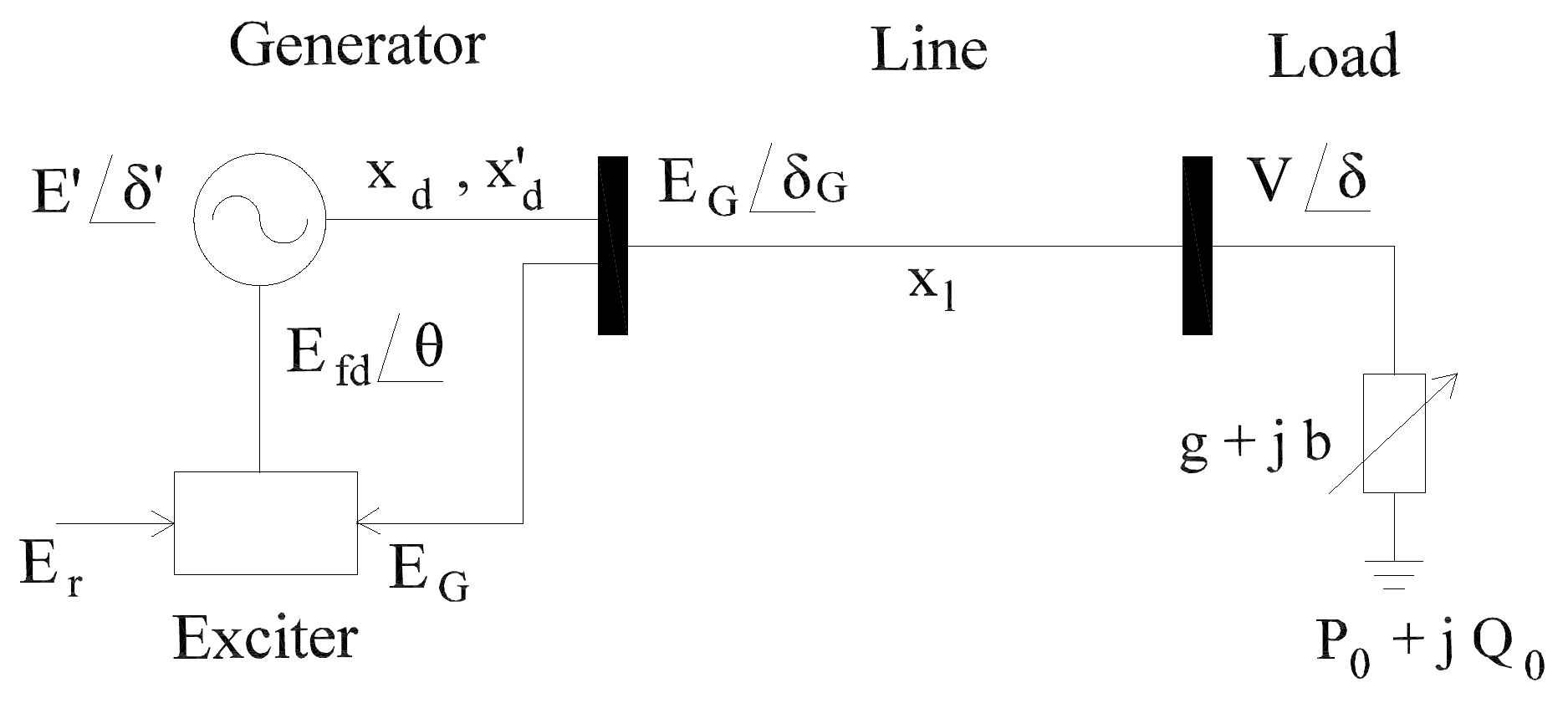}
	\caption{A rudimentary system \cite{Venkatasubramanian1992}}
     \label{fig:rudi}
\end{figure}

The rudimentary 2-bus system shown in Figure \ref{fig:rudi} is adopted from \cite{Venkatasubramanian1992} and is extended with the dynamic model of the loads. The generator consists of an internal voltage source behind the transient reactance and an IEEE Type $1$ exciter. In this work, we do not consider angle dynamics but focus solely on voltage dynamics, although the extension to more general models is trivial. The set of differential equations describing the generator dynamics are the same as described in \cite{Venkatasubramanian1992} or \cite{PaiDynamics}:
\begin{equation} \label{eq:Ep}
\centering
T'_{d0}\dot{E'} = -\frac{x_d}{x'_d} E'+\frac{x_d-x'_d}{x'_d}\,E_{G}\,\cos(\delta_{G}-\delta')+E_{fd}
\end{equation}
\begin{equation} \label{eq:Efd}
\centering
T\dot{E}_{fd} = -E_{fd} - K(E_{G}-E_r)
\end{equation}
where $x_d$ and $x'_d$ are the equivalent direct axis reactance and transient direct axis reactance; $T'_{d0}$ is the direct axis transient open circuit time
constant; $E'\angle \delta'$ is the internal source voltage; $E_G\angle \delta{_G}$ is the terminal voltage; $E_R$ is the reference voltage; $E_{fd}$ is the exciter output voltage (generator field voltage); $K$ and $T$ are the gain and integral time constant associated with exciter PI control. Generator models are described in details in \cite{Kundur, Heydtreport, exciterreport}.
The dynamics of the load is described by \eqref{eq:geqsimple}:

\begin{align}\label{eq:rudiload}
\tau \dot{g} = - (g V^2 - P_0) = -(p-P_0), \\
\tau \dot{b} = - (b V^2 - Q_0) = -(q-Q_0).
\end{align}
where $\tau$ is the load time constant, $\tau=\tau_{g}=\tau_b$; $V$ is the voltage magnitude at the load bus; $P_0=P^S$ and $Q_0=Q^S$ are the desired demand levels that we assume to be constant and not depending on the voltage; $p$ and $q$ are the instantaneous power consumptions of the load. For the rudimentary system, the set of state variables includes $4$ states, i.e. $x=[E',E_{fd},g,b]^\top$ which can be decomposed into $2$ state vectors $x_\mathcal{G}=[E',E_{fd}]^\top$ and $x_\mathcal{L}=[g,b]^\top$. Moreover, the diagonal matrix constituted by the time constants of the loads is $\mathcal{T}=\mathrm{diag}(\tau,\tau)$. The relations (\ref{eq:Efd}) and (\ref{eq:rudiload}) form the set of differential equations in (\ref{eq:DAEs}).

Algebraic equations, $G(x,y) = 0$ are composed of relation describing the generator, the network, and the load can be stated as follow:
\begin{align}
0 =& \frac{E'E_{G}}{x'_d}\sin(\delta{_{G}}-\delta')+\frac{E_{G}V}{x_{l}}\sin(\delta_{G}-\delta)\\
0 =& \frac{1}{x'_d}(E_{G}^2-E_{G}E'\cos(\delta{_{G}}-\delta{'}))\\
&+\frac{1}{x_{l}}(E_{G}^2-E_{G}E\cos(\delta{_{G}}-\delta))\\
0 =& \frac{VE_{G}}{x_{l}}\sin(\delta-\delta{_{G}})+p\\
0 =& \frac{1}{x_{l}}(V^2-E_{G}E\cos(\delta-\delta{_{G}}))+q\\
p =& gV^2\\
q =& bV^2
\end{align}

The internal voltage source angle is used as the reference, i.e. $\delta'=0$. The system parameters are given as the following: $T'_{d0}=5$; $E_r=1$; $x_d=1.2$; $x'_d=0.2$; $T=0.39$; $K=10$; $x_{l}=0.1$. All parameters are in $p.u.$ except time constants in second and scalar gain $K$.

\begin{figure}[!ht]
    \centering
    \includegraphics[width=\columnwidth]{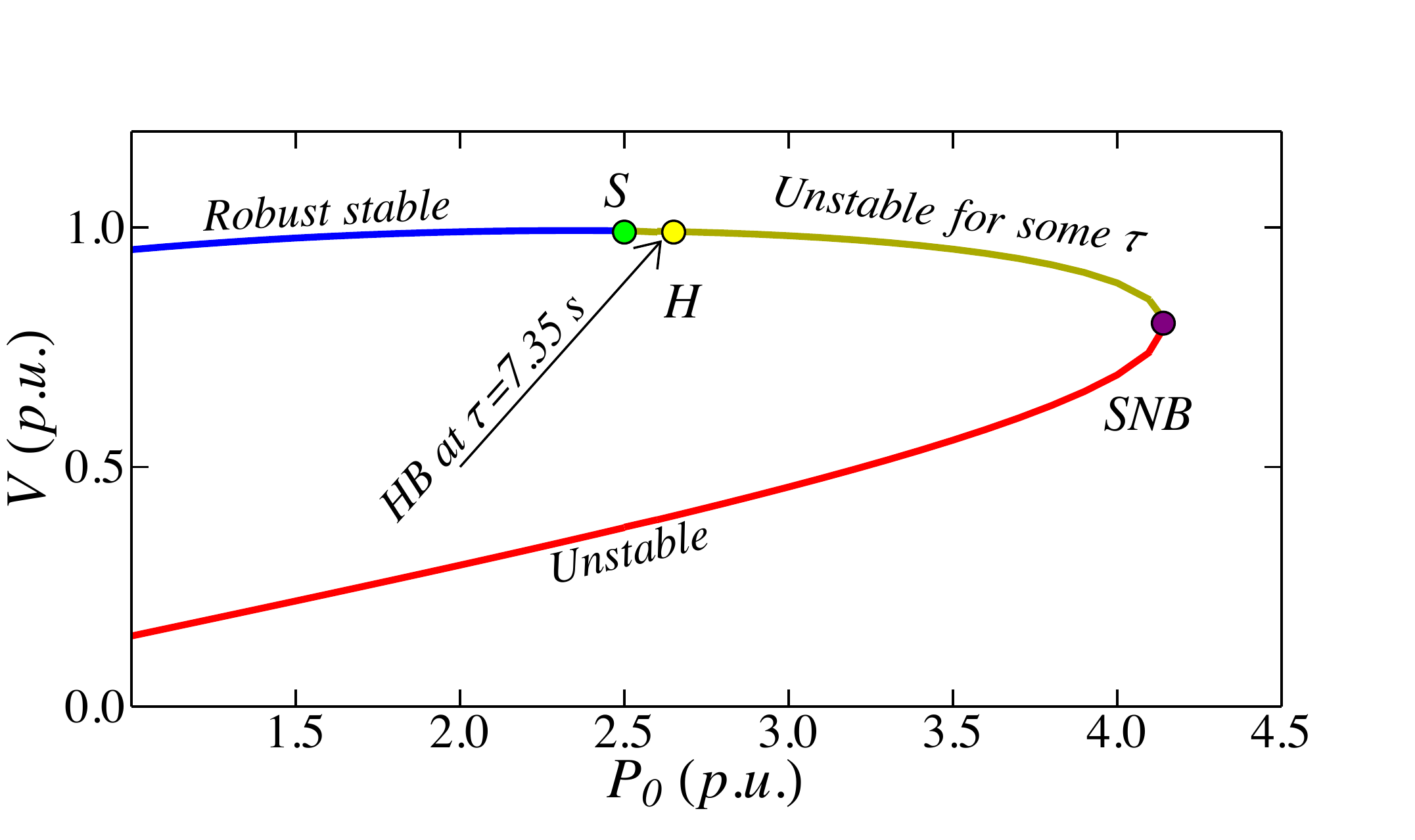}
	\caption{Robust stability illustration for rudimentary system}
     \label{fig:rudiPV}
\end{figure}

In Figure \ref{fig:rudiPV} we show the results of stability analysis of different points on the nose curve. The system is shown to be robust stable up to point $S$ where $P_0=2.51\, p.u.$ at the upper branch of the nose curve of $\cos\phi=0.98$. Saddle-node bifurcation ($SNB$) corresponding to voltage collapse occurs at $P_0=4.2\,p.u.$. The section of the upper branch between $S$ and $SNB$ cannot be certified to be robust stable, and can be numerically shown to be unstable for some load time constant $\tau$ at every point. For example, at point $H$ where $P_0=2.6 \, p.u.$, the system exhibits Hopf bifurcation ($HB$) with $\tau=7.35\,s$. The eigenvalues of matrix $A$ at point $H$ are shown on figure \ref{fig:rudiEig}.

For the rudimentary system, the lower branch of the $PV$ is unstable for most of load dynamics.

\begin{figure}[!ht]
    \centering
    \includegraphics[width=0.8 \columnwidth]{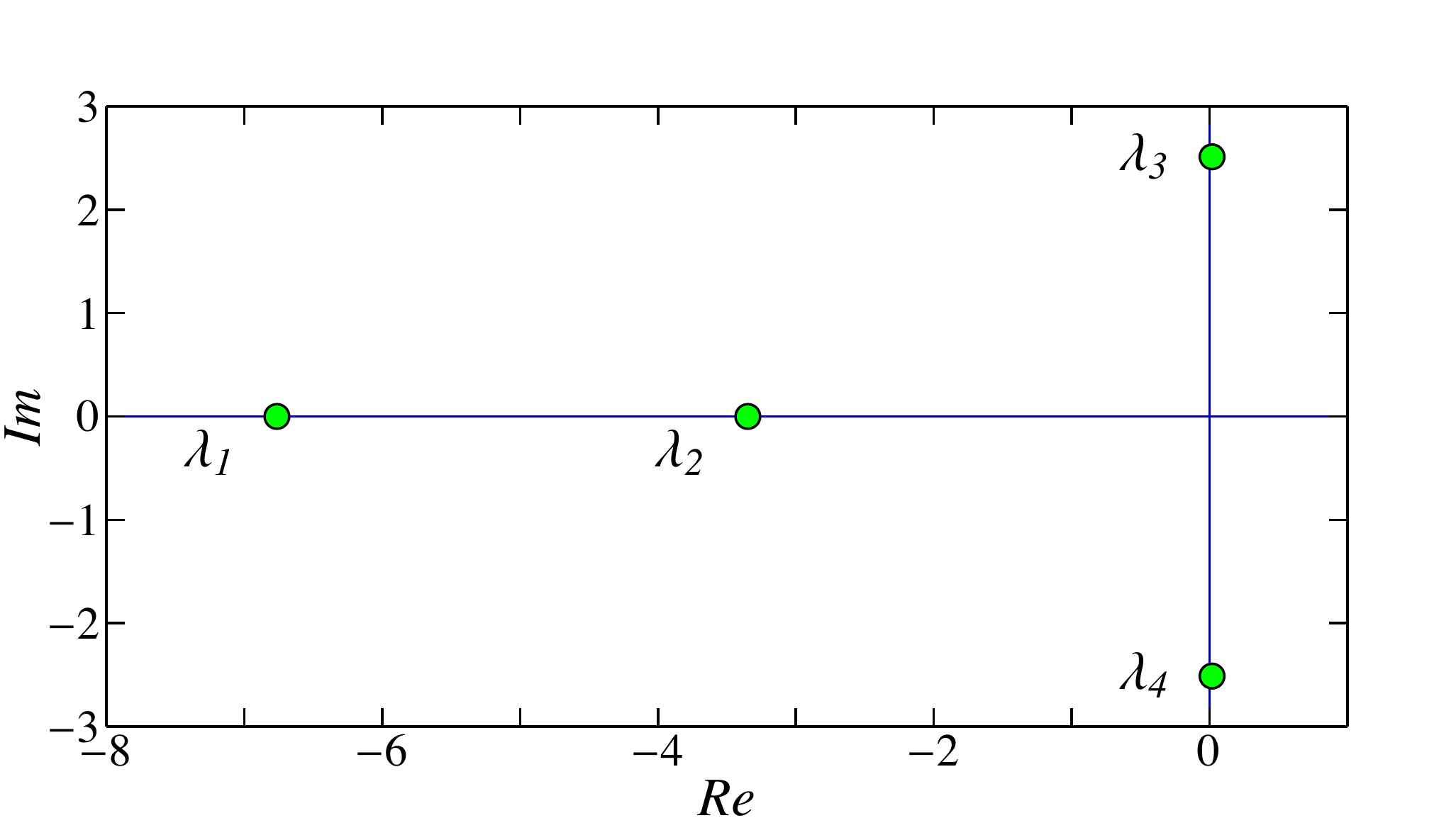}
	\caption{The eigenvalues of $A$ matrix of rudimentary system encountering Hopf bifurcation}
     \label{fig:rudiEig}
\end{figure}

\subsection{The WSCC $3$-machine, $9$-bus system} \label{sec:WSCC}

\begin{figure}[!ht]
    \centering
    \includegraphics[width=1 \columnwidth]{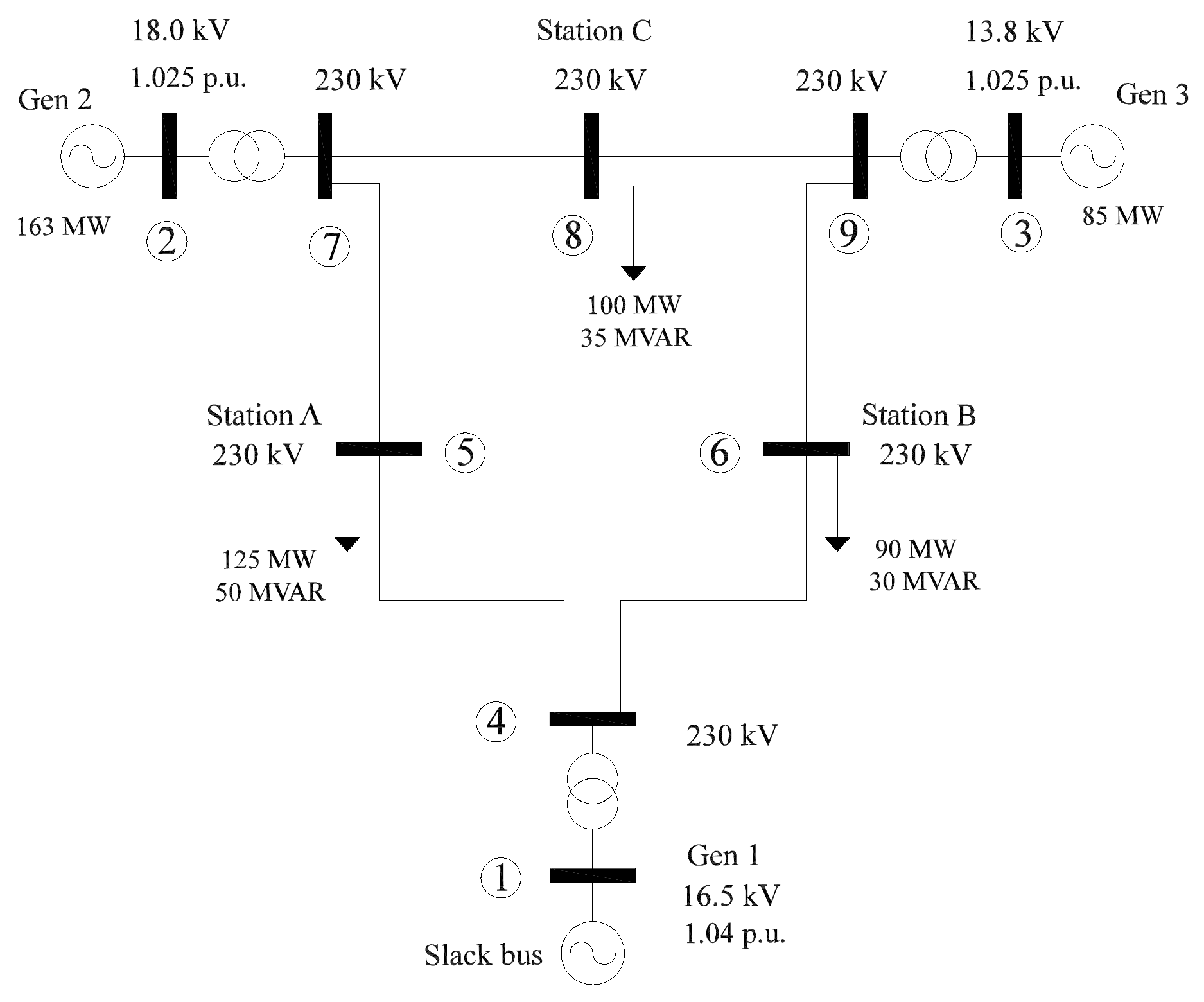}
	\caption{The WSCC $3$-machine, $9$ bus system \cite{PaiDynamics}}
     \label{fig:WSCC}
\end{figure}
The WSCC $3$-machine, $9$-bus system with all the parameters is plotted in Figure \ref{fig:WSCC}. Bus $1$ is the slack bus, and bus $2$ and $3$ are $PV$ buses with specified the active power outputs and the magnitude of voltages at the terminals. Three $PQ$ loads are connected to $3$ substations residing at buses $5$, $6$, and $8$. The base power is $S_{base}=100\,MVA$. We assume that load bus $8$ works with a constant power factor, i.e. $\cos\phi_8=0.894$. All branches and transformers data are described in Appendix \ref{app:9busdata}.

To characterize the stability of the system we increase the load at bus $8$ while keeping the other parameters fixed. The system is robust stable up to point $S$ where $P_8=3.0\,p.u.$. The region from $S$ to SNB where saddle-node bifurcation happens at $P_8=3.5\,p.u.$, the system may become unstable for some time constants. For example, fixed time constant of load $5$ and $6$ to be equal $1\,s$, the system encounters Hopf bifurcation at point $H_1$ where $P_8=3.36\,p.u.$, $\tau_8=15.57\,s$, or at point $H_2$ where $P_8=3.45\,p.u.$, $\tau_8=11\,s$.

\begin{figure}[!ht]
    \centering
    \includegraphics[width=1 \columnwidth]{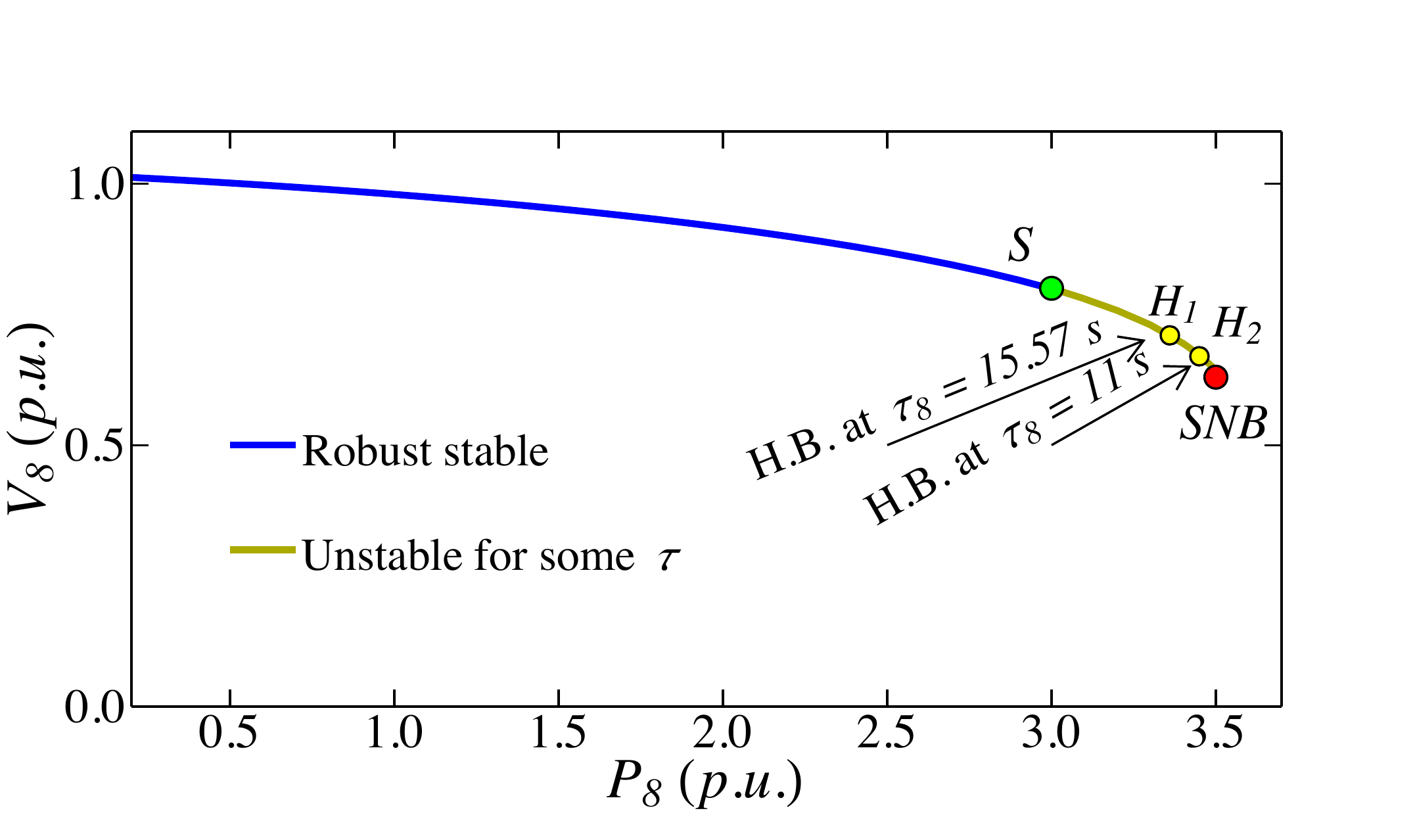}
	\caption{Robust stability illustration for WSCC $3$-machine, $9$-bus system}
     \label{fig:WSCCPV}
\end{figure}

\begin{figure}[!ht]
    \centering
    \includegraphics[width=1 \columnwidth]{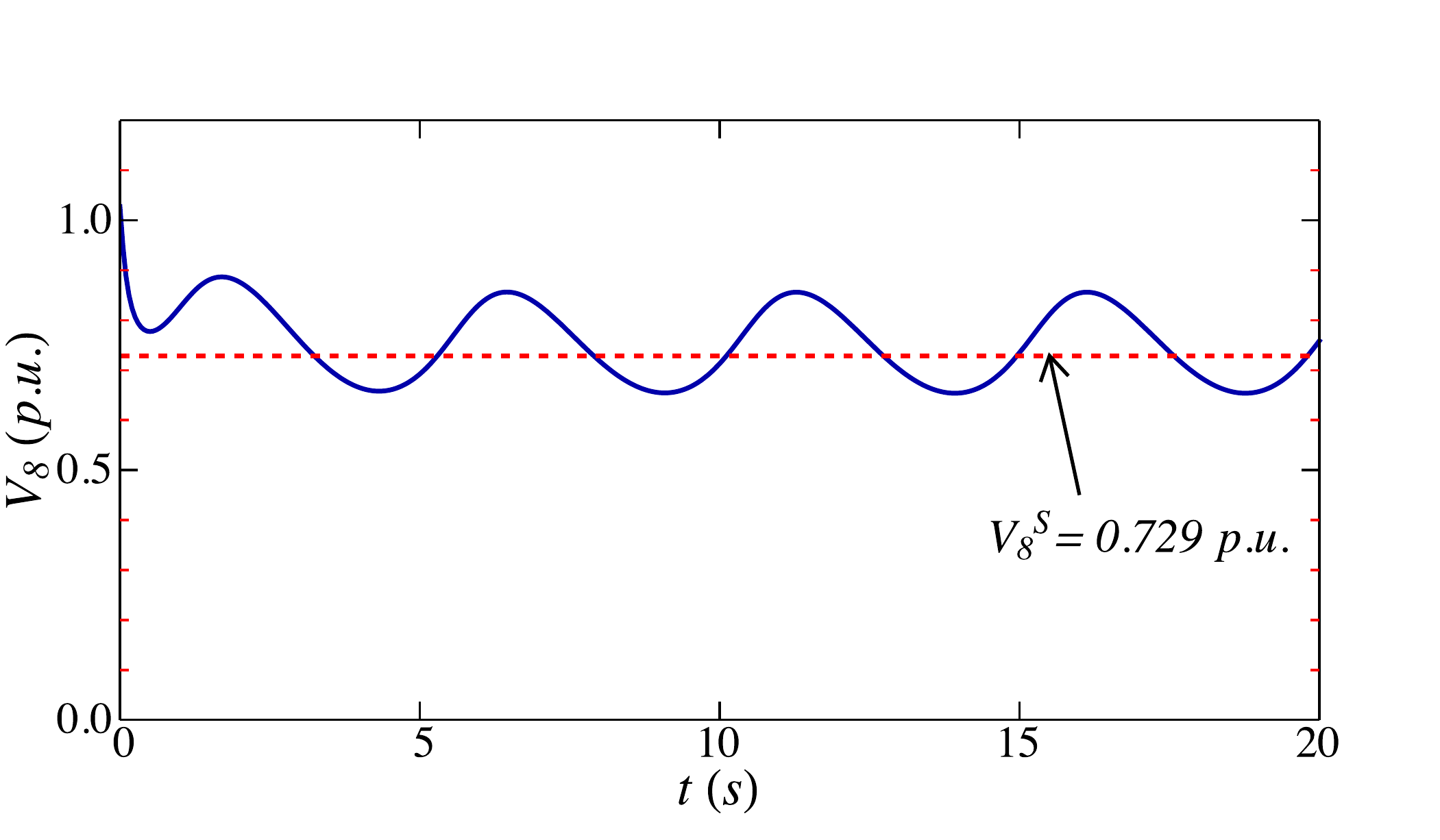}
	\caption{Oscillatory voltage instability with the WSCC $3$-machine, $9$-bus system at $H_2$ where $P_8=3.45\,p.u.$ and $\tau_8=11\,s$}
     \label{fig:WSCCH2}
\end{figure}
In Figure \ref{fig:WSCCPV}, $V_8^S$ is the voltage level when the system is stable for the same level of power consumption, i.e. $P_8=3.45\,s$ but with smaller time constant, say $\tau_8=9\,s$. For less uncertain systems, i.e. when load buses $5$ and $6$ have fixed $\tau_g=\tau_b$, point S may extent to higher level of active power at bus $8$, $P_8=3.1\,p.u.$. This observation is true for more general situations, i.e. the less uncertainty presents in the system, the more stable the system is. 

\begin{figure}[!ht]
    \centering
    \includegraphics[width=1 \columnwidth]{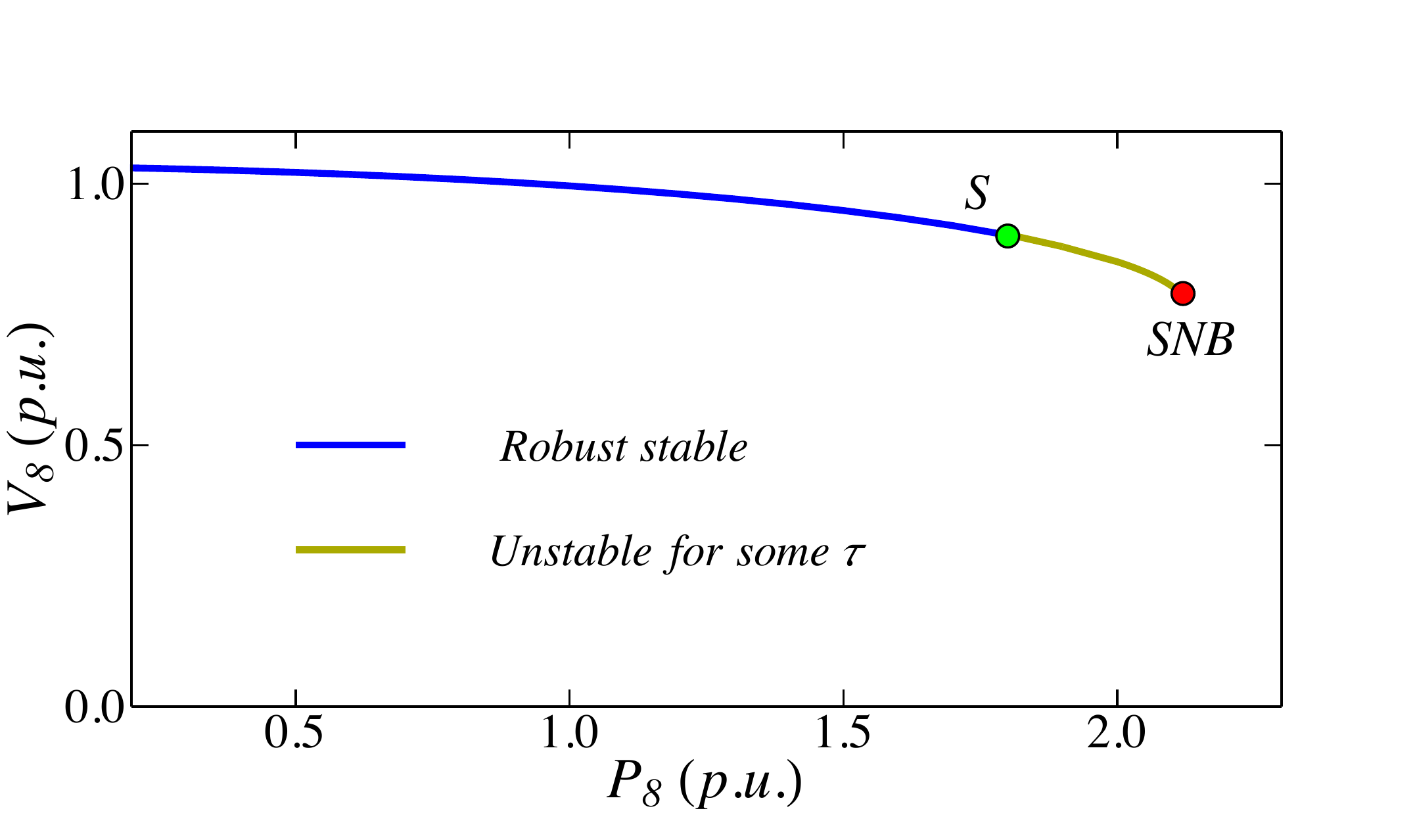}
	\caption{Robust stability illustration for WSCC $3$-machine, $9$-bus system, correlated loading condition}
     \label{fig:WSCCPV_loading}
\end{figure}

Also, we consider a more realistic loading scenario with correlated loading condition. We consider the case when $P_5=P_6=P_8$ and $Q_5=Q_6=Q_8$. Again, the $PV$ curve shown in Figure \ref{fig:WSCCPV_loading} indicates the robust stability region in blue where $P_8\leq1.86\,p.u.$ and the yellow region, from point S to SNB, where the system may become unstable for some instant relaxation times of the loads. Figure \ref{fig:WSCCPV_loading} resembles Figure \ref{fig:WSCCPV} where no correlated loading scenario is considered. They differ only in loading conditions at the robust stable point, S, and the saddle-node bifurcation. The lower critical loading conditions are observed because the power transferred through power lines increases faster when all buses are loaded at once. Different correlated loading scenarios considered but not reported in the manuscript were characterized by qualitatively similar results as shown in either Figure \ref{fig:WSCCPV} or Figure \ref{fig:WSCCPV_loading}. In the follow-up section \ref{39bus} we also report similar studies with more realistic economic load dispatch scheme that accounts for distribution of the load increase between different generators \cite{PaiDynamics}. The behavior observed in that scenario is also qualitatively similar.

\subsubsection{RSA for WSCC $3$-machine $9$-bus system} \label{sec:VSAexample}
As mentioned before, in this subsection we demonstrate the application of robust stability applied to RSA within $N-1$ security assessment. Different from off-line assessment in which an exhaustive list of contingencies is assessed, here we only consider a set of most dangerous contingencies. This practice, indeed, is more suitable for online assessment. The subset of considered contingencies may include the lines with large power flows or the lines that are connected to low voltage buses \cite{LiuVSA2000}. 
The base case power flow is chosen as shown in Figure \ref{fig:WSCC} except for load bus $8$, where $P_8 = 1.8 \,p.u.$, $Q_8=0.5\,p.u.$. For the WSCC $3$-machine $9$-bus system, all the voltage levels are close to $1\,p.u.$. Therefore, we rely on the total MVA power flows through the line to determine the most dangerous ones.

\begin{table}[ht]
    \caption{Contingency analysis summary table}
    \centering
    \label{table:contingency}
\begin{tabular}{|l||*{4}{c|}}\hline
\textbf{Line trip}
&\makebox[3em]{$1-4$}&\makebox[3em]{$2-7$}&\makebox[3em]{$7-8$}&\makebox[3em]{$9-3$}
\\\hline\hline
Case \rom{1} &Stable&Stable&Stable&Stable\\\hline
Case \rom{2} &Limit Cycle&Stable&Stable&Stable\\\hline
Case \rom{3} &Unstable&Unstable&Limit Cycle&Stable\\\hline\hline
\textbf{RSA} &NRS&NRS&NRS&RS\\\hline
\end{tabular}
\end{table}

There are two different situations in contingency analysis, i.e. with uncertainty or without uncertainty. When there is no uncertainty in the model,  consider $3$ different cases of fixed time constants at bus $5$, $6$, and $8$; i.e. $\tau_{5}=\tau_{6}=\tau_{8}=\tau$, and $\tau=1\,s$ in Case \rom{1}, $\tau=5\,s$ in Case \rom{2}, $\tau=10\,s$ in Case \rom{3}. The absolute values of the instant relaxation time are not important because the actual set of the time constants of the loads may vary over time and may be different from bus to bus. Therefore, the $3$ cases are used merely to demonstrate the performance of robust stability analysis. In contrast, we use RSA in the presence of uncertainty. For each dangerous contingency and such time constants, the system stability is assessed as shown in Table \ref{table:contingency}. 
\begin{figure}[htb]
    \centering
    \subfigure[Trip line $1-4$, Limit Cycle]{{\includegraphics[width=4.5cm]{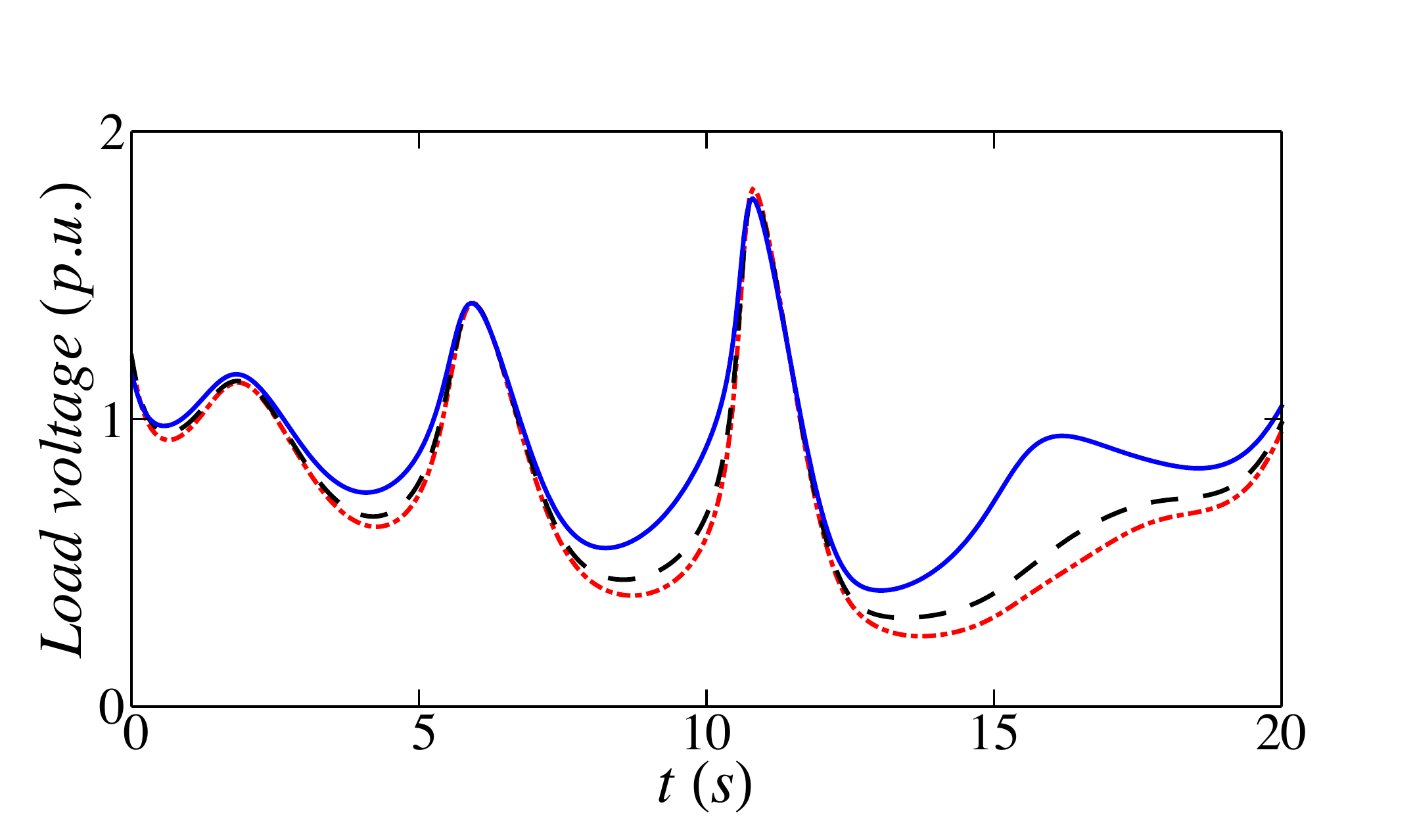}}}%
    \subfigure[Trip line $2-7$, Stable]{{\includegraphics[width=4.5cm]{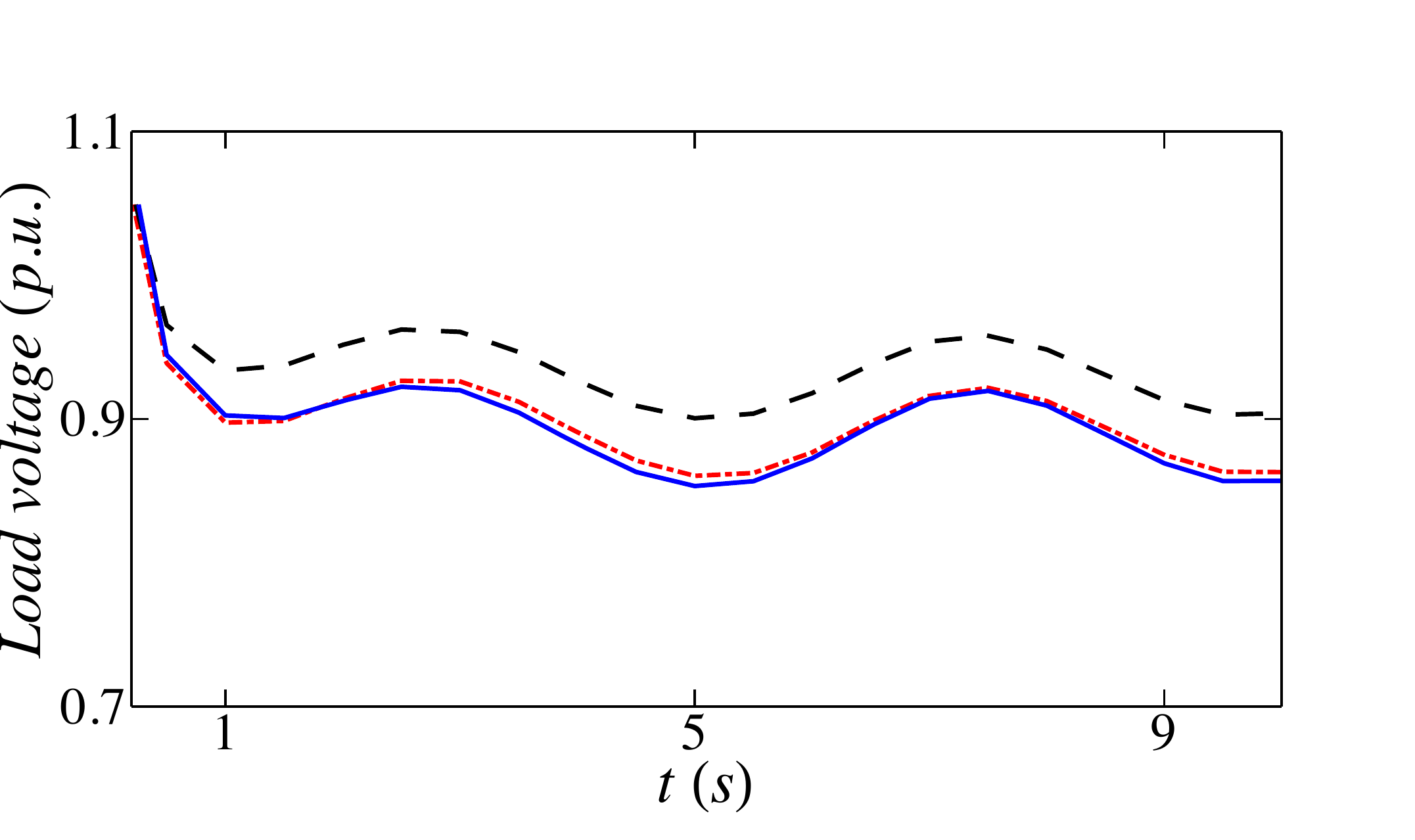} }}%
    \quad
    \subfigure[Trip line $7-8$, Stable]{{\includegraphics[width=4.5cm]{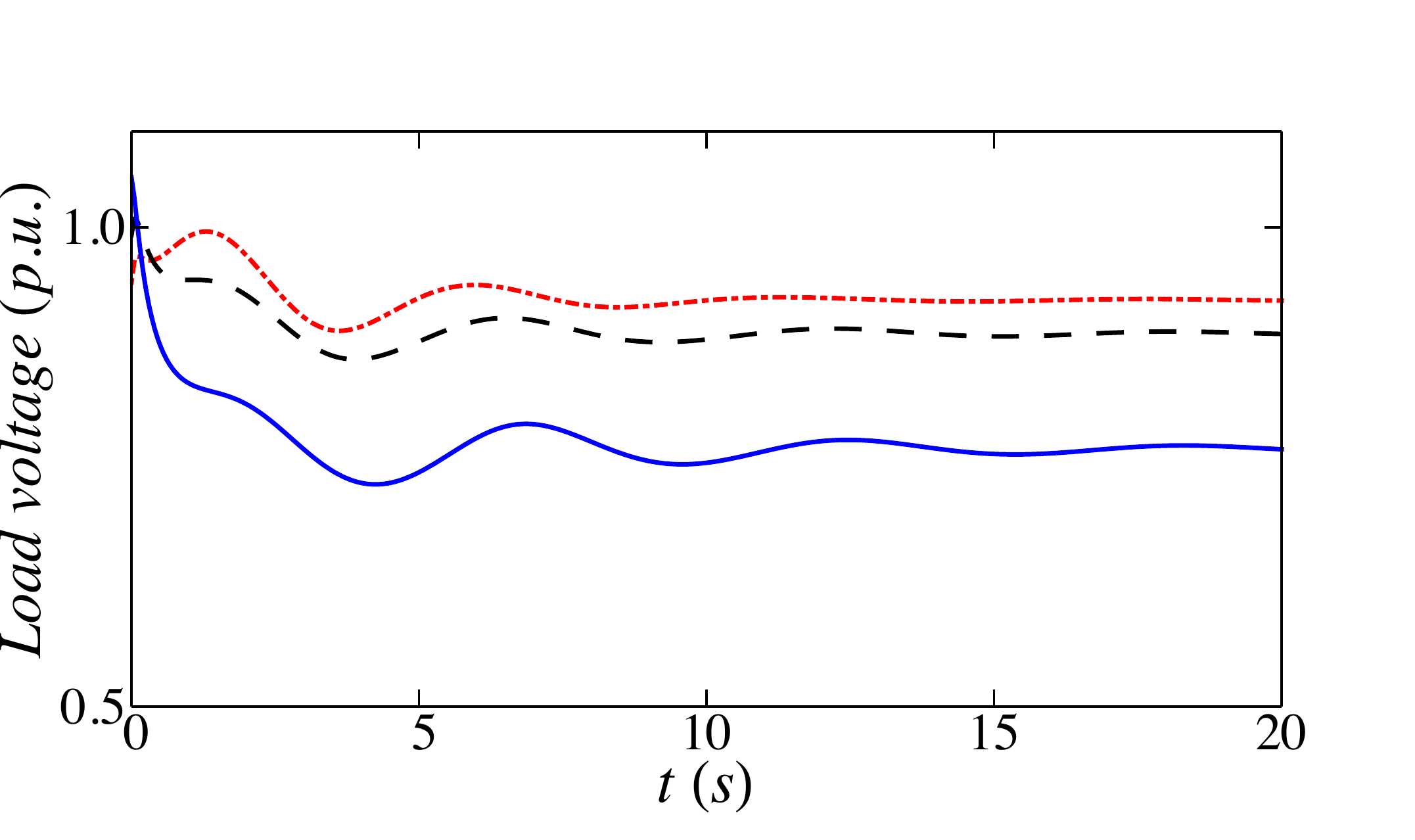}}}%
    \subfigure[Trip line $9-3$, Stable]{{\includegraphics[width=4.5cm]{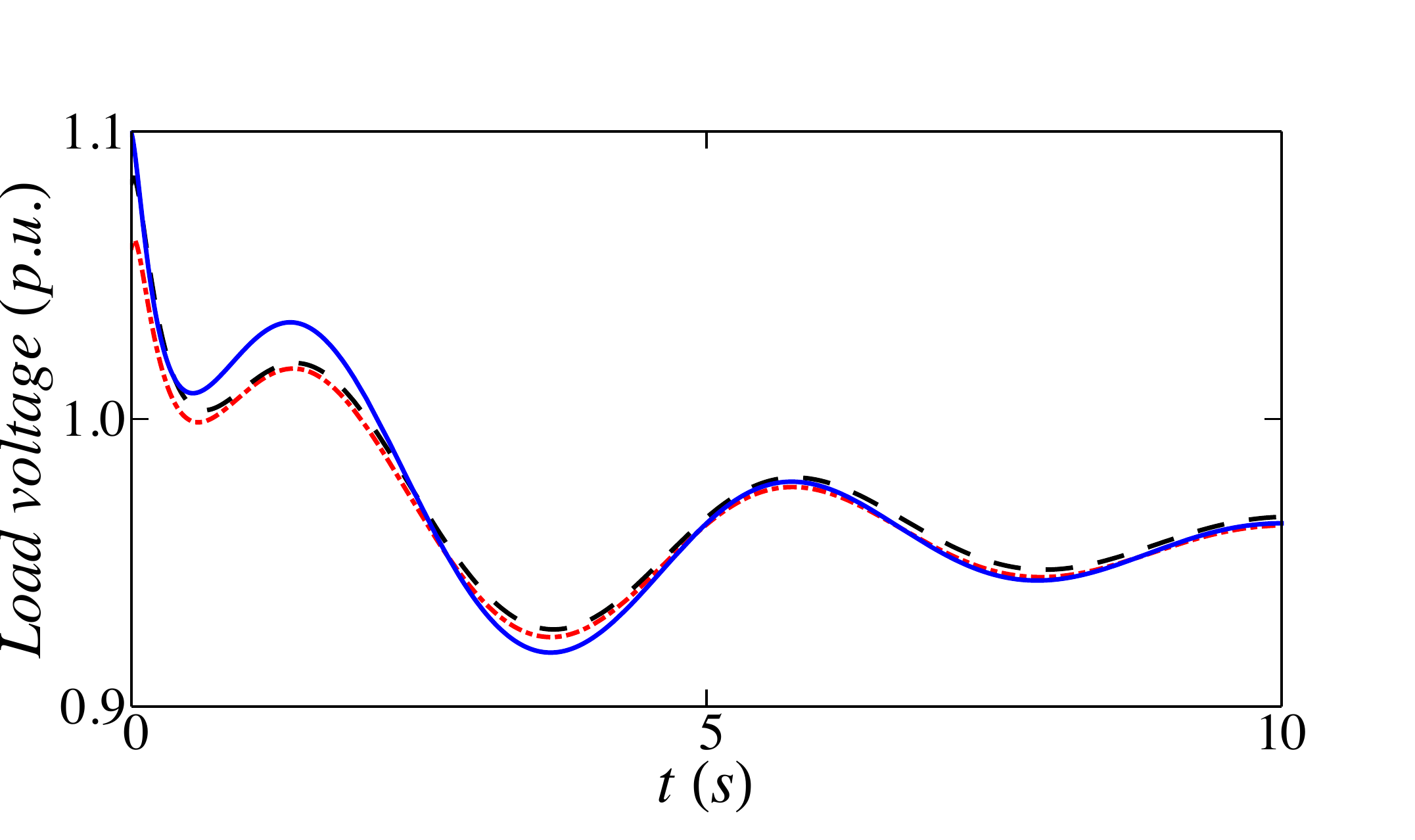} }}%
    \caption{The load voltage evolutions in time-domain simulations in contingency analysis for Case \rom{2}, $\tau=5\,s$}%
    \label{fig:case2}%
\end{figure}

\begin{figure}[htb]
    \centering
    \subfigure[Trip line $1-4$, Unstable]{{\includegraphics[width=4.5cm]{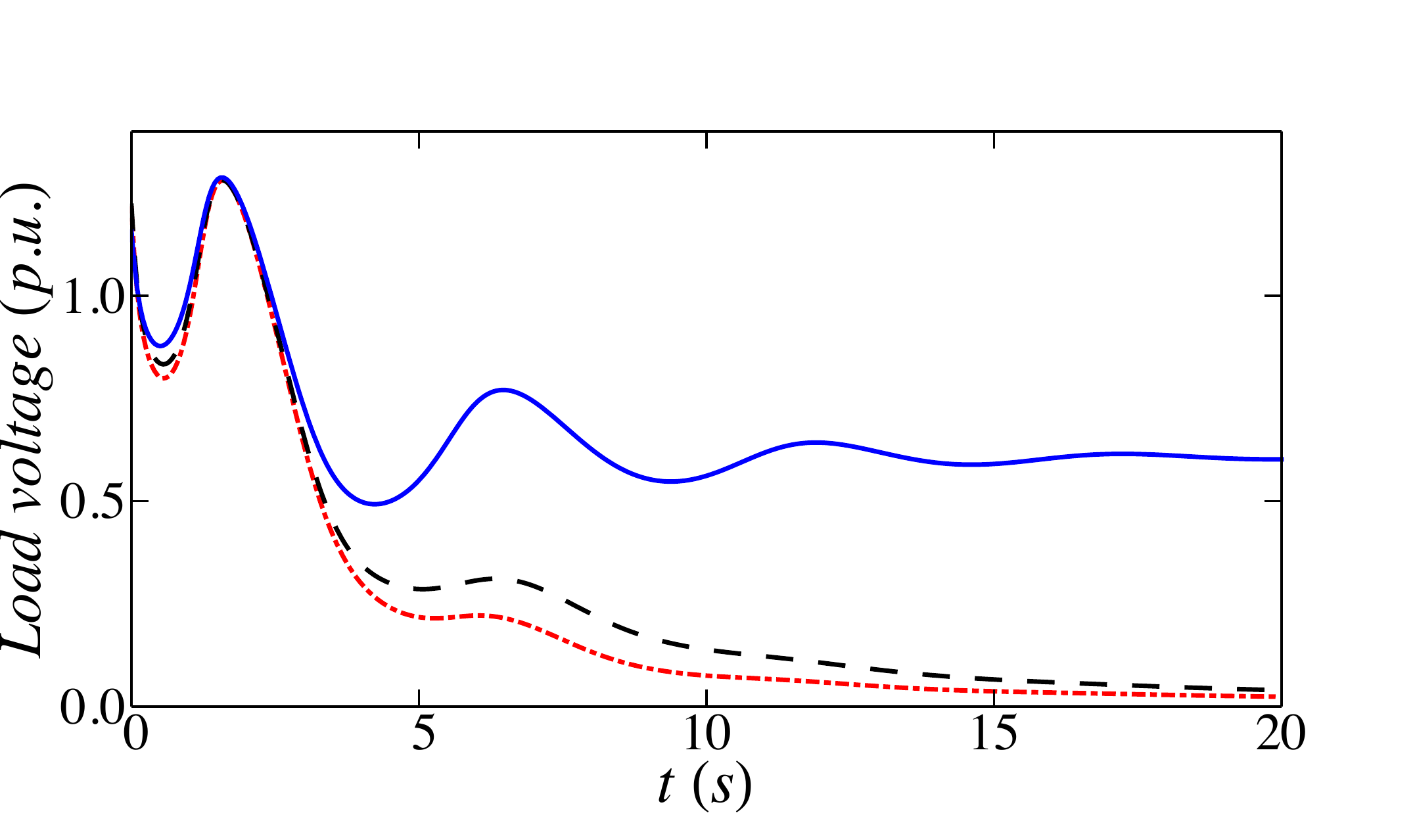}}}%
    \subfigure[Trip line $2-7$, Unstable]{{\includegraphics[width=4.5cm]{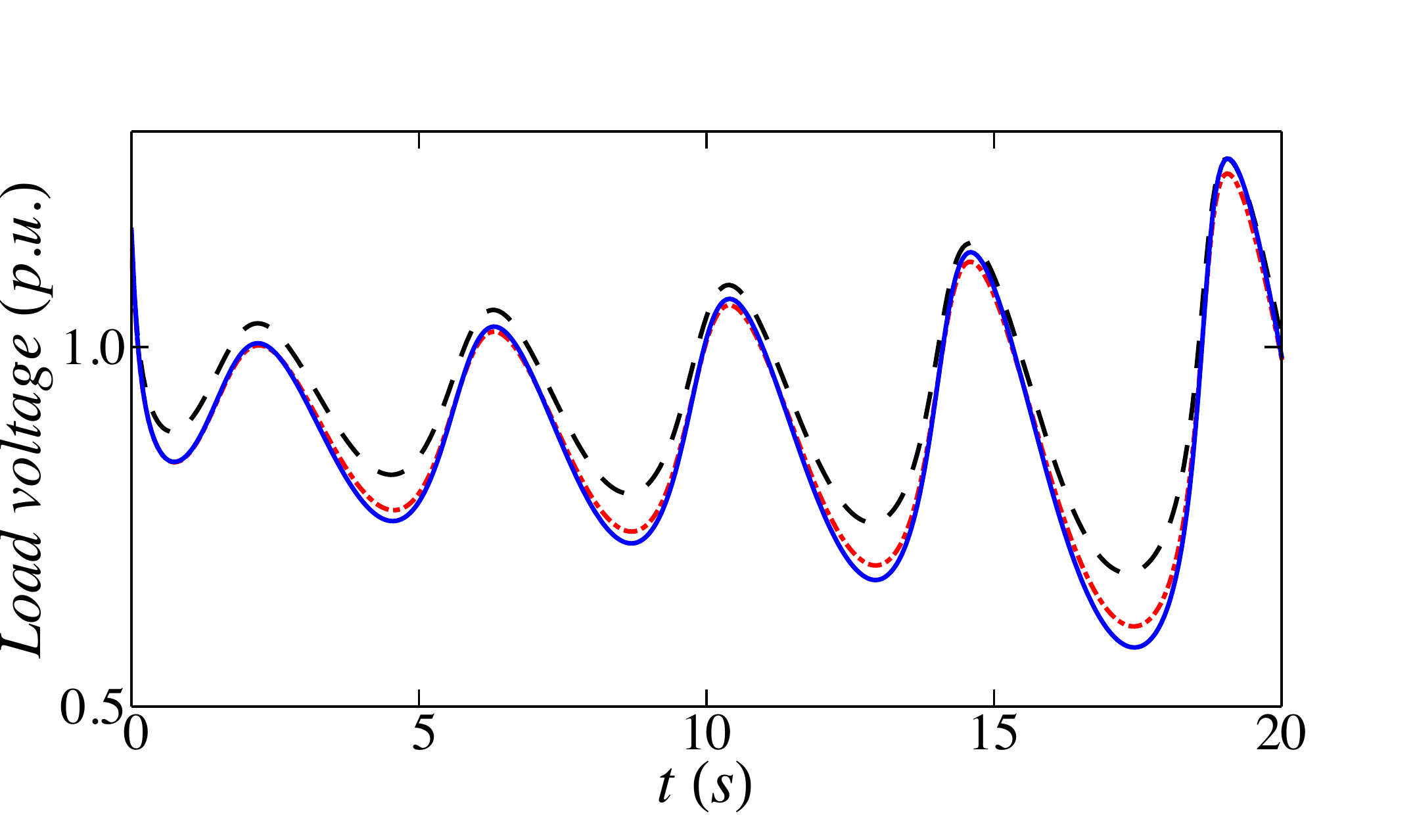} }}%
    \quad
    \subfigure[Trip line $7-8$, Limit Cycle]{{\includegraphics[width=4.5cm]{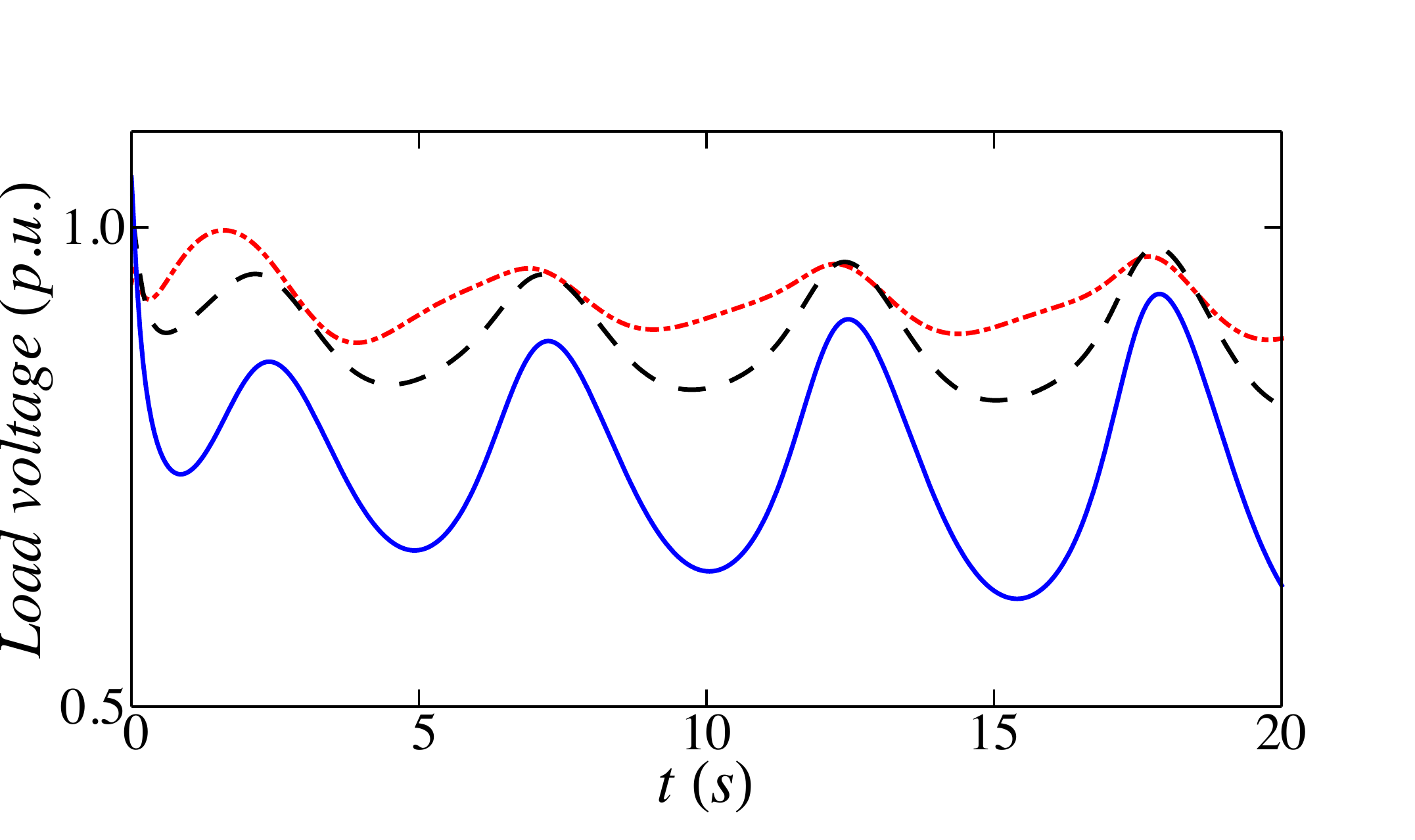}}}%
    \subfigure[Trip line $9-3$, Stable]{{\includegraphics[width=4.5cm]{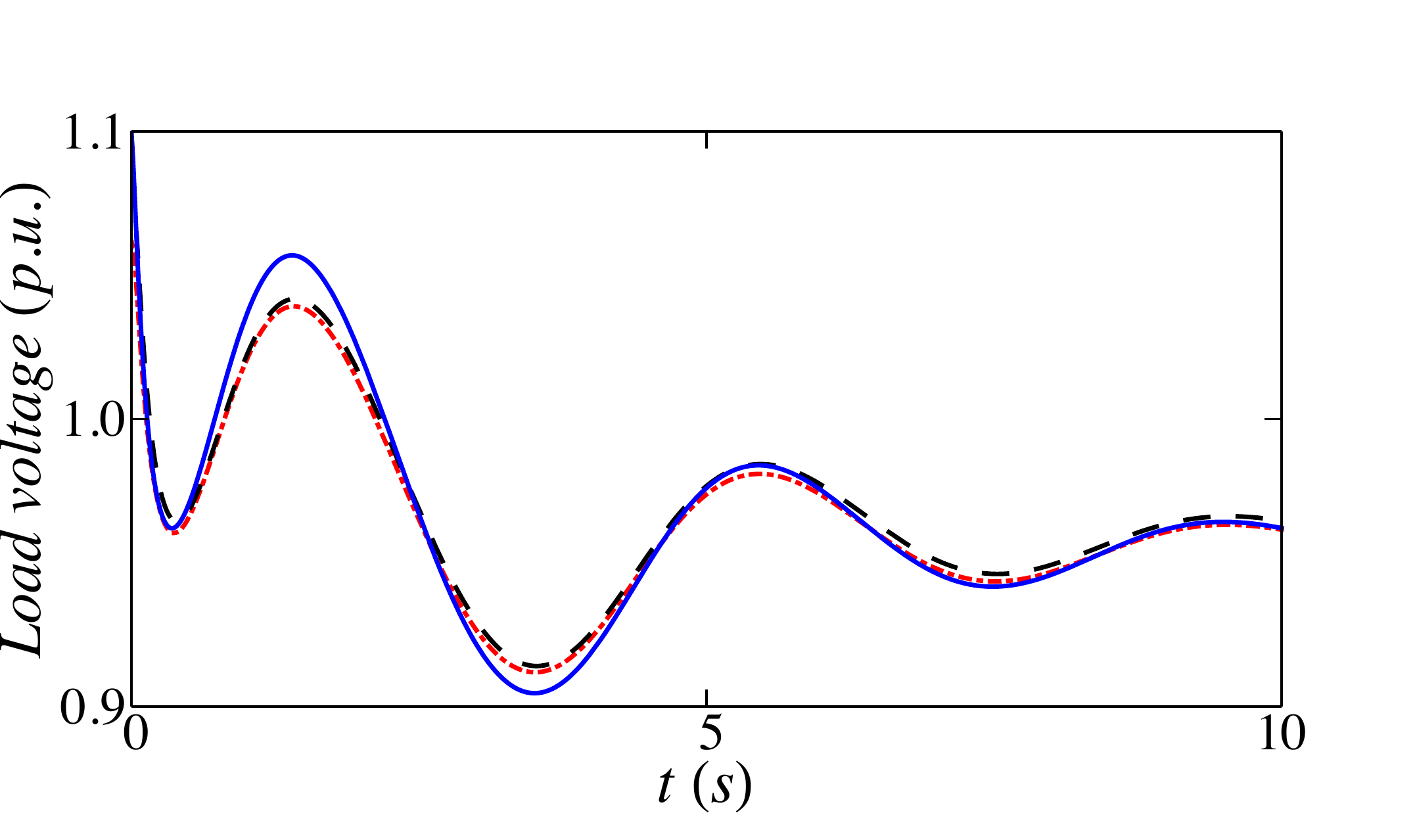} }}%
    \caption{The load voltage evolutions in time-domain simulations in contingency analysis for Case \rom{3}, $\tau=10\,s$}%
    \label{fig:case3}%
\end{figure}

In Table \ref{table:contingency}, for RSA results, RS and NRS imply robust stable and non robust stable, respectively. One can observe that if the system is robust stable, for example when line $9-3$ is tripped, the non-uncertain stability assessment also indicates that the system is stable in all cases. In contrast, if the system is not robust stable according to RSA results, there exists some cases or some set of instant relaxation times cause the system unstable. This happens when either line $1-4$ or $2-7$ is disconnected. Moreover, in two considered cases, the system is stable if the line $7-8$ is tripped. For this contingency, RSA result indicates that the system is non-robust stable. In fact, the system is unstable with $\tau_5=\tau_6=1 \,s$ and $\tau_8>14\,s$ where the load voltage at bus $8$ collapses around $t=60\,s$. 

In considered situations, limit cycles (LC) appear in Case \rom{2} with line $1-4$ tripping and in Case \rom{3} with line $7-8$ tripping. The system will exhibit voltage oscillations which are unexpected and dangerous because they may limit the power transfers and induce stress in the mechanical shafts \cite{PaiDynamics}. In such cases, RSA also indicates that the system is non-robust stable or potentially unstable.

The contingency analysis results, for example in Case \rom{2} and Case \rom{3}, can also be represented with time-domain simulations as in Figure \ref{fig:case2} and Figure \ref{fig:case3} where the red dash-dot, black dash, and blue solid trajectories correspond to the load voltages at bus $5$, $6$, and $8$, respectively. For $\tau=5\,s$ and  tripping the line $1-4$, the system encounters Hopf bifurcation and the voltages keep oscillating but never go beyond the range from $0.2\,p.u.$ to $1.8\,p.u.$. Also, for $\tau=5\,s$ and  tripping the line $2-7$, the system is stable but very lightly damped. The voltages settle around $t=800\,s$ which indicates that the system is close to Hopf bifurcation point. The first $20$-second and $10$-second evolutions of the load bus voltages when tripping the line $1-4$ and $2-7$ for Case \rom{2} are presented in Figure \ref{fig:case2}(a) and Figure \ref{fig:case2}(b), respectively. Moreover, for Case \rom{3}, the line $2-7$ is tripped, the voltage at the load bus $8$ collapses around $t=80\,s$; hence the system is unstable. Figure \ref{fig:case3}(b) shows the first $20$-second time evolution of the unstable voltage trajectory.

However, RSA does not require any time-domain simulation, thus reduces the need of storages and the time consuming. In addition, RSA does not provide the margin to SNB or particular bifurcation points, instead RSA provides another type of stability margin i.e. robust stability margin which measures the distance between the current operating point to the robust stability boundary. For example, for the contingency case in which the line $9-3$ is tripped, the security indicator discussed in section \ref{sec:SI}, $SI=\rho=0.004$, indicates that the system will work close to the robust stability boundary after the contingency. Hence, a slight change in parameters will cause the system move to the non-robust stable region where it may become unstable. In contrast, the contingency cases with the line $2-7$ tripping, even though the system is non-robust stable, the security SI is very small, i.e. $SI=\rho=-3.4\times10^{-5}$. If appropriate control is applied, the system will be secure in the robust stability region. In this sense, RSA with SI can help the system operators in designing emergency controls.

As aforementioned, it may be impossible to determine the actual values of the instant relaxation times of the loads. Without making any assumption about the load responses, RSA is recommended to run first to screen the most dangerous contingency set. If the RSA certifies that the system is robust stable, no further action is needed; otherwise, deeper analysis or other probabilistic-based assessments such as Monte Carlo simulations are required. Therefore, if RSA is used as the very first screening, the whole process of contingency analysis is expedited.

\subsection{IEEE $39$-bus New England system} \label{39bus}

\begin{figure}[!ht]
    \centering
    \includegraphics[width=.9 \columnwidth]{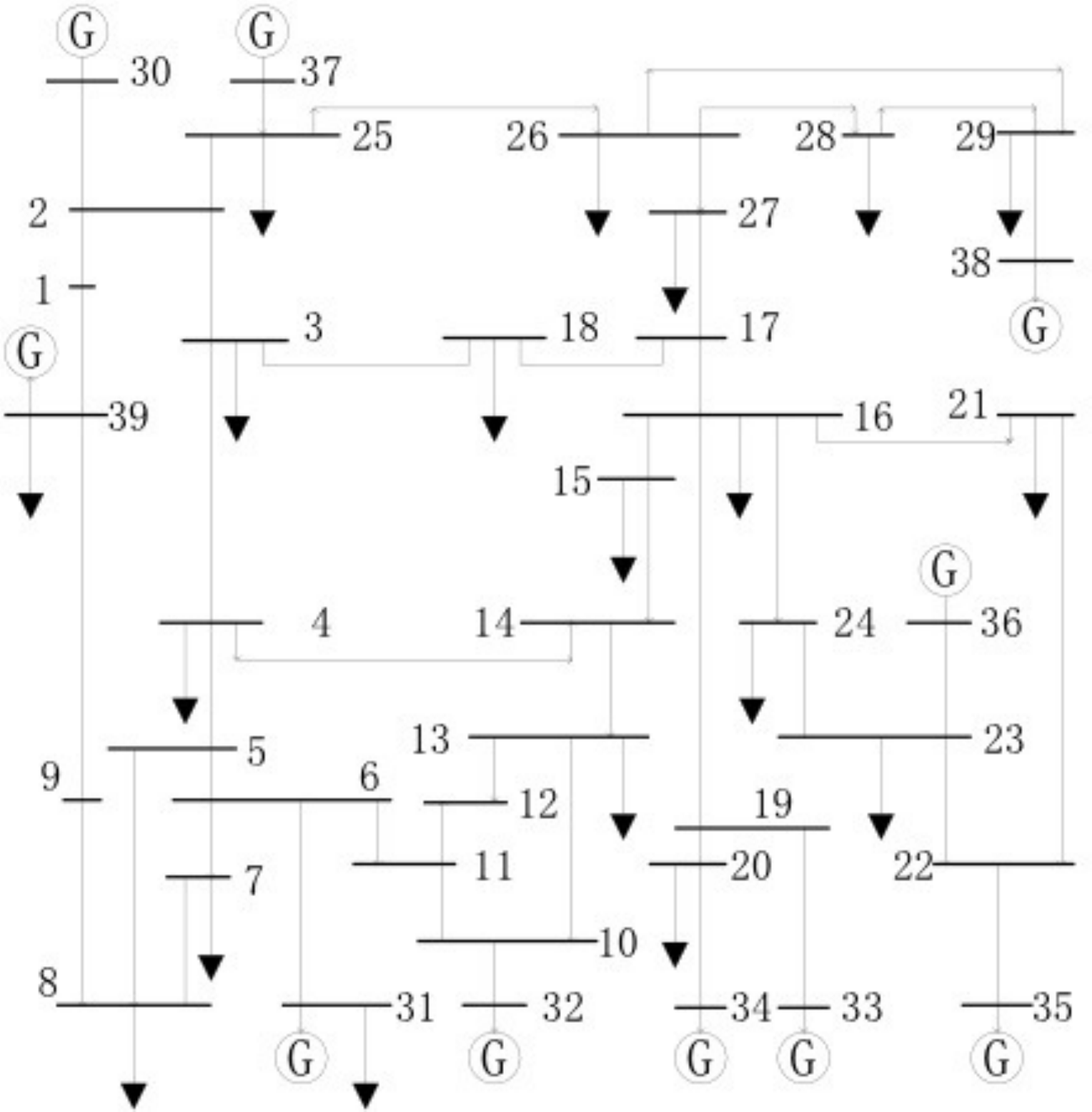}
	\caption{The New England system}
     \label{fig:NE}
\end{figure}

In this section, we  illustrate the concept of robust stability with the IEEE $39$-bus New England system. The configuration of the system is shown in Figure \ref{fig:NE}. All generators are identical and have the same set of parameters as the following: $T'_{d0}=10\,s$; $x_d=1.0\,p.u.$; $x'_d=0.2\,p.u.$; $T=0.39\,s$; $K=10$. Other system parameters are adopted from \cite{39busNE}. In the considered scenario, all the loads have the same power factor, i.e. $\cos(\phi)=0.9$ lagging; the load bus $29$ is chosen as the reference load and other load levels are increased with the correlated loading factor $k_c$, i.e. $P_i=k_cP_{29}$, where $i \in \mathcal{L},\,i \neq 29$. We will consider the situation with identical load power consumptions or $k_c=1$. The load increments were picked up by evenly distribution among all generators.

\begin{figure}[!ht]
    \centering
    \includegraphics[width=1 \columnwidth]{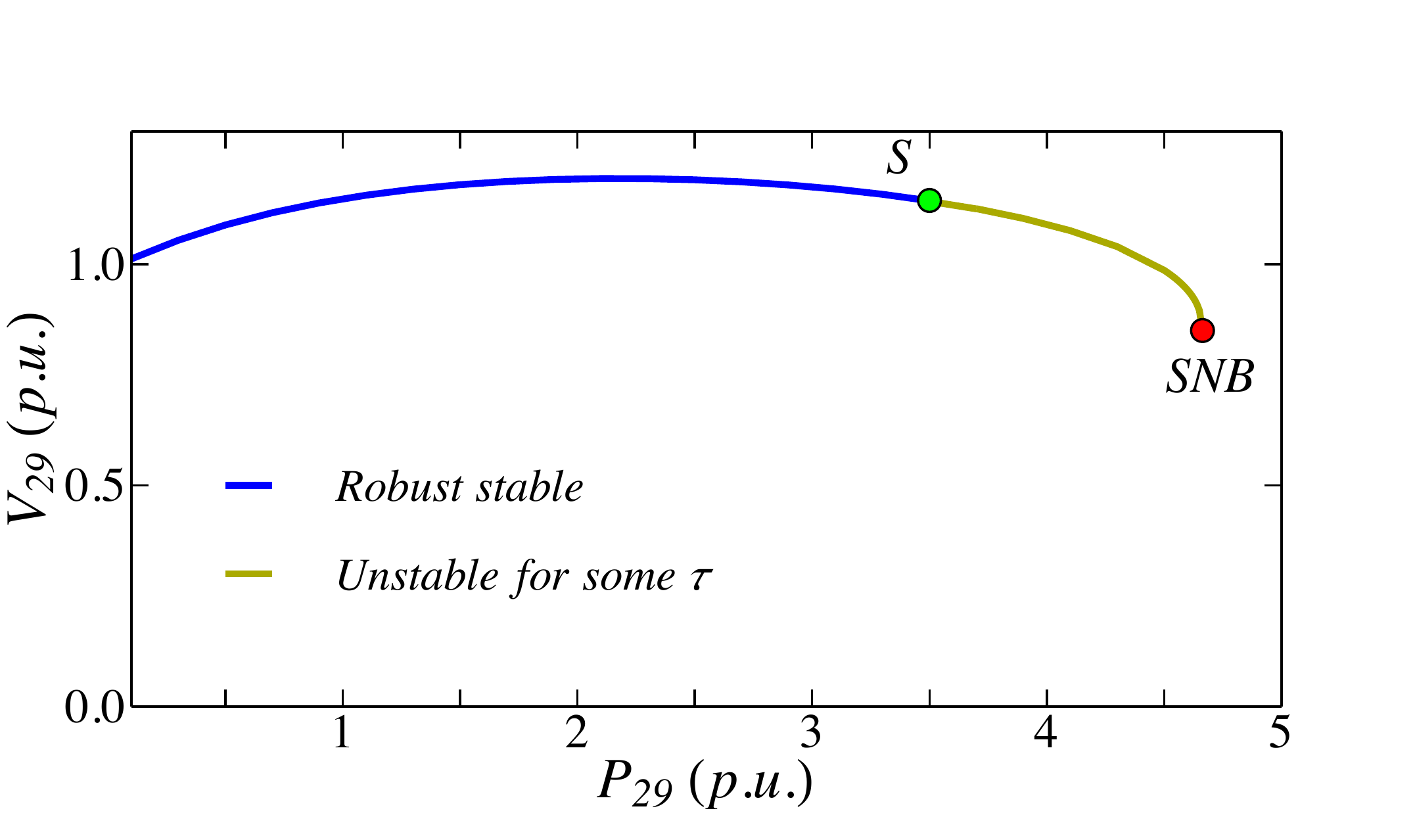}
	\caption{Robust stability illustration for the New England system, correlated loading condition $k_c=1$}
     \label{fig:NEPV_loading}
\end{figure}

For the given scenario, the robust stability of the New England system is illustrated in Figure \ref{fig:NEPV_loading} which is similar to that of the rudimentary system and WSCC $3$-machine, $9$-bus system. The system is robust stable up to point S where $P_{29}=3.5\,p.u.$. SNB occurs near $P_{29}=4.67\,p.u.$. Therefore, the margin from S to SNB is around $25.05\%$.

We also considered another loading scenario where the base loading levels are adopted from \cite{39busNE}. Then for each load the power factor is kept unchanged while all the load consumptions are scaled with the same scalar factor $k_c>0$. In this scenario, SNB happens at $k_c=3.0$ and the system is robust stable up to $k_c=1.2$. This means that the system can become unstable at some loading level that is above $20\%$ of the normal operating condition. Moreover, the margin from S to SNB is $60\%$.

\section{Investigation of the non-certified robust stability region} \label{sec:nonrobust}

In Figure \ref{fig:rudiPV}, \ref{fig:WSCCPV}, \ref{fig:WSCCPV_loading}, and \ref{fig:NEPV_loading}, the non-certified robust stability regions are in yellow and lie between the robust stable point S and the saddle-node bifurcation point SNB. Different from the robust stability region, the non-robust stability one is mostly affected by the load dynamic uncertainty. The system dynamics and behavior may be very different and complicated because of more pronounced nonlinearity. When the system is stressed or is subject to disturbances, the system is likely to operate in those regions. Therefore, it is important to explore the non-robust stability regions which may help the system operators to have better understanding of the system. We will address two important questions in this section, i.e. which parameter determines the robust stable point S and how the system behaves in the non-certified robust stability region.

\subsection{Robust stable point S}
The position of point S as well as the robust stability region characterizes the level of ``robustness'' of the system. For the same configuration, the size of robust stability region might vary from case to case, from scenario to scenario.

\subsubsection{Effect of loading levels}
We reconsider the scenario with correlated loading condition, i.e. $P_5=P_6=k_c\,P_8$ and $Q_5=Q_6=k_c\,Q_8$ where $k_c$ is the correlation factor. Table \ref{table:kc} illustrates how the system loading levels affect the robust stability region. The margin in $\%$ measures the distance between point S and SNB compares to the maximum loading level corresponding to SNB.

\begin{table}[ht]
    \caption{Effect of loading levels on S}
    \centering
    \label{table:kc}
\begin{tabular}{|l||*{4}{c|}}\hline
\textbf{$k_c$}
&\makebox[3em]{$0.5$}&\makebox[3em]{$1$}&\makebox[3em]{$2$}&\makebox[3em]{$4$}
\\\hline\hline
\textbf{S} ($p.u.$) &2.70&1.86&1.07&0.55\\\hline
\textbf{SNB} ($p.u.$) &3.10&2.16 &1.22&0.65\\\hline
\textbf{Margin} ($\%$) &12.90&13.89 &12.30&15.38\\\hline
\end{tabular}
\end{table}

From Table \ref{table:kc}, one can see that an increase in the correlation loading factor resulted in an decrease in the maximum loading level where SNB happens. However, increasing $k_c$ may not necessarily lead to the change in the robust stable point S in such a way that extends the margin between S and SNB.

\subsubsection{Effect of load power factors}
Various power factors were considered in Table \ref{table:pf}. One can see that as the load power factors change from lagging to leading, the relative distance between the robust stable point S and SNB increases. This means that the more lagging the power factor is, the wider the robust stable region becomes. Therefore, injecting more reactive powers into the network may shorten the robust stability region relatively.

\begin{table}[ht]
    \caption{Effect of power factor on S}
    \centering
    \label{table:pf}
\begin{tabular}{|l||*{5}{c|}}\hline
\textbf{power factor}
&\makebox[3em]{$0.5\, lag$}&\makebox[3em]{$0.9\, lag$}&\makebox[3em]{$1.0$}&\makebox[3em]{$0.9\, lead$}&\makebox[3em]{$0.5\, lead$}
\\\hline\hline
\textbf{S} ($p.u.$) &0.95&1.86&2.30&2.40&2.20\\\hline
\textbf{SNB} ($p.u.$) &1.00&2.16&2.74&3.35&4.80\\\hline
\textbf{Margin} ($\%$) &5.00&13.89 &16.06&28.36&54.17\\\hline
\end{tabular} 
\end{table}

\subsubsection{Effect of exciter gain $K$}
The model of exciter is described in \eqref{eq:Efd}. In this section, effect of exciter gain $K$ is analyzed in Table \ref{table:K}. As observed in \cite{PaiDynamics}, the sufficient increase of the exciter gain may lead to instability even for normal loading level. With robust stability analysis, we now can determine at which loading level the exciter gain cannot affect the system stability by considering $K$ as an uncertain parameter.

\begin{table}[ht]
    \caption{Effect of exciter gain $K$ on S}
    \centering
    \label{table:K}
\begin{tabular}{|l||*{6}{c|}}\hline
\textbf{$K$}&\makebox[2.5em]{$5$}
&\makebox[2.5em]{$10$}&\makebox[2.5em]{$20$}&\makebox[2.5em]{$30$}&\makebox[2.5em]{$40$}&\makebox[2.5em]{$50$}
\\\hline\hline
\textbf{S} ($p.u.$) &1.60&1.80&1.86&1.87&1.96&1.97\\\hline
\textbf{SNB} ($p.u.$) &2.16&2.16&2.16&2.16&2.16&2.16\\\hline
\textbf{Margin} ($\%$) &25.93&16.67&13.89 &13.43&9.26&8.79\\\hline
\end{tabular}
\end{table}
As expected, the changing in $K$ does not affect the maximum loading level at SNB point. However, surprisingly, an increase in $K$ tends to extend the robust stable region as pushing point S closer to SNB point. When $K$ goes to infinity, point S does not change much and the system is robust stable up to circa $P_8=2.00\,p.u.$. This indicates that exciter gain may affect the system stability in a rather complicated manner which depends on the interactions between exciters and generators with other dynamic devices/components; as well as depends on the considered conditions/scenarios.

\subsection{The system behavior in the region between S and SNB}

Since dynamic voltage stability is normally studied by monitoring the eigenvalues of the linearized system \cite{PaiDynamics}, we investigate how these factors alter the system eigenvalues in the s-plane. The rudimentary system results are demonstrated as below.

\subsubsection{Effect of loading levels}

\begin{figure}[!ht]
    \centering
    \includegraphics[width=1 \columnwidth]{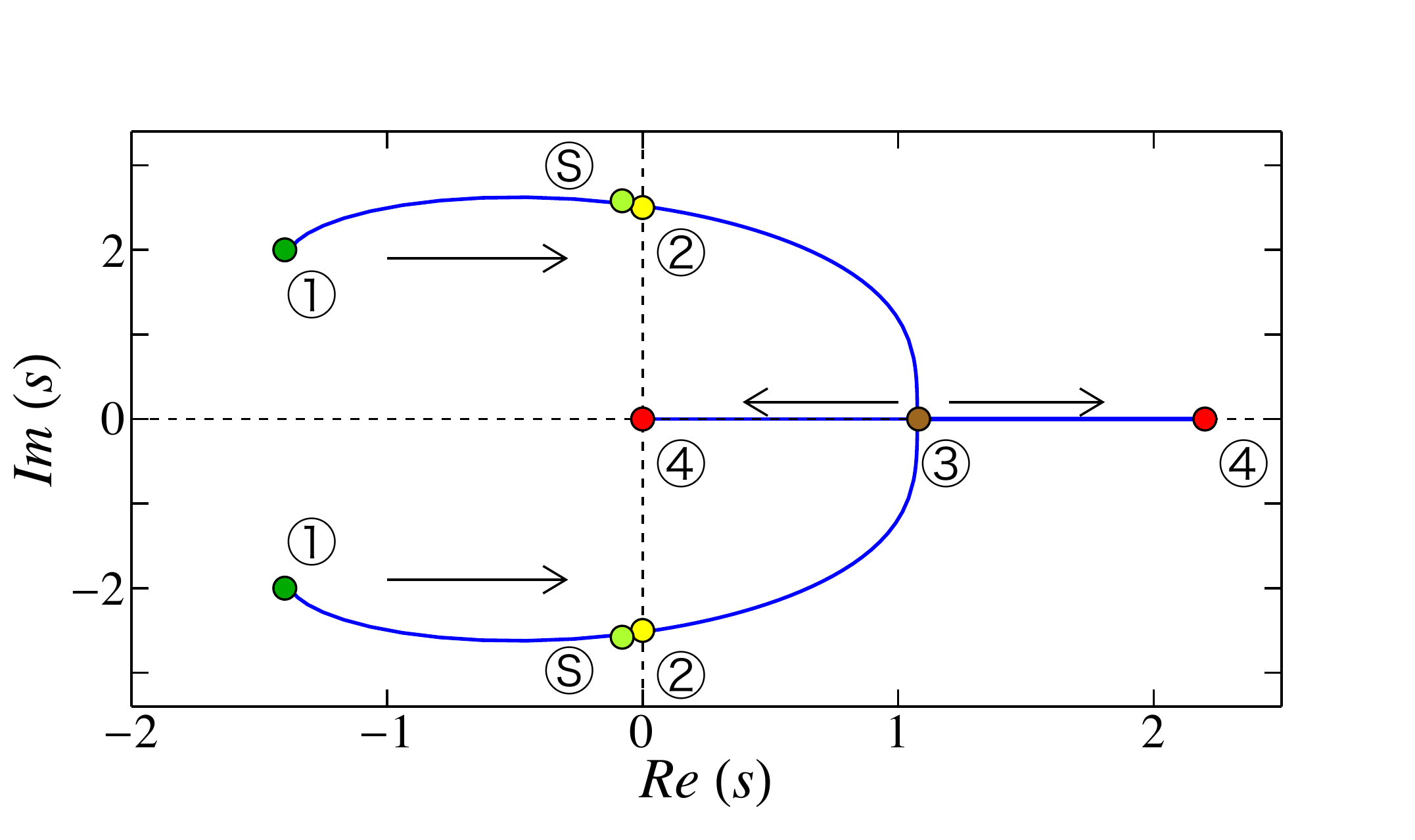}
	\caption{Critical eigenvalue trajectory under the load changes in the rudimentary system, $\tau=7.35\,s$}
     \label{fig:crieigentraj_2bus_tau735}
\end{figure}

For $\tau=7.35\,s$, the trajectory of the critical eigenvalue pair, \circled{1}-\circled{S}-\circled{2}-\circled{3}-\circled{4}, is plotted in Figure \ref{fig:crieigentraj_2bus_tau735} as the load power increases from zero to the maximum loading level. Note, that the enclosed alphanumerics indicate that the corresponding eigenvalues belong to the same system matrix which is related to the same power level consumption $P_0$. In Figure \ref{fig:crieigentraj_2bus_tau735}, the pair of critical eigenvalues starts at \circled{1} with zero power level consumption and move to the right half plane in the s-plane. When the trajectory crosses the imaginary axis at \circled{2} where $P_0=2.6\,p.u.$, the system encounters Hopf bifurcation. This is also illustrated at point H in Figure \ref{fig:rudiPV}. The eigenvalues associated with the power level at robust stable point S in RSA are marked with \circled{S} which is close to \circled{2}. As the load power continues increasing, the two critical complex eigenvalues coalesce at \circled{3} on the real axis of the s-plane and become a pair of real eigenvalues. Then the pair of critical real eigenvalues diverge following the two arrows towards \circled{4}. As soon as the one that moves to the left reaches \circled{4} at the origin, the SNB occurs. Since the load power cannot exceed the maximum loading level, the trajectory ends here at \circled{4}. The similar trajectory is also described in \cite{PaiDynamics}.

\begin{figure}[ht]
    \centering
    \includegraphics[width=1 \columnwidth]{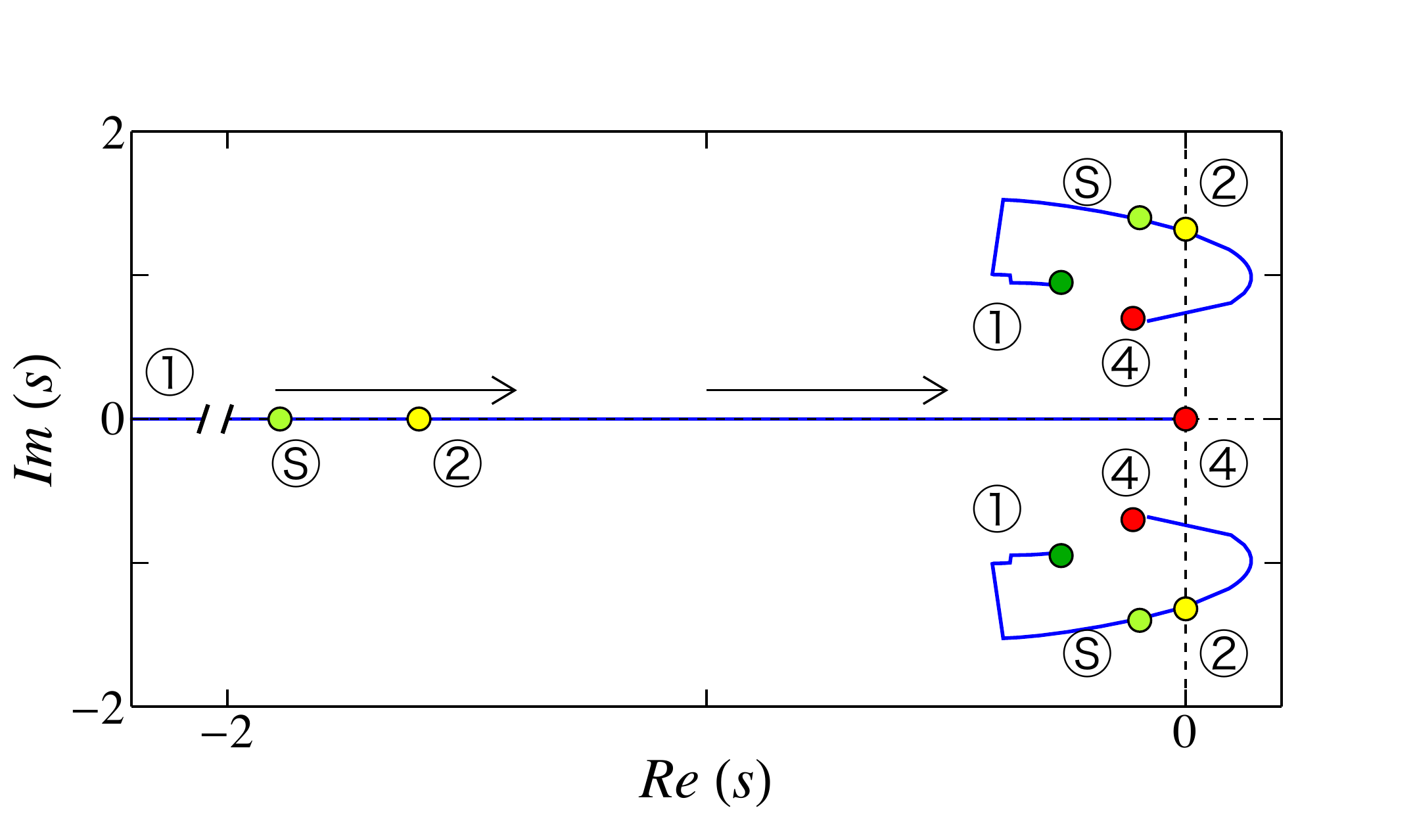}
	\caption{Critical eigenvalue trajectory under the load changes in the WSCC $3$-machine, $9$-bus system}
     \label{fig:crieigentraj_9bus}
\end{figure}

\begin{figure}[htb]
    \centering
    \subfigure[{\ding[1.3]{192}} Stable] {{\includegraphics[width=4.5cm]{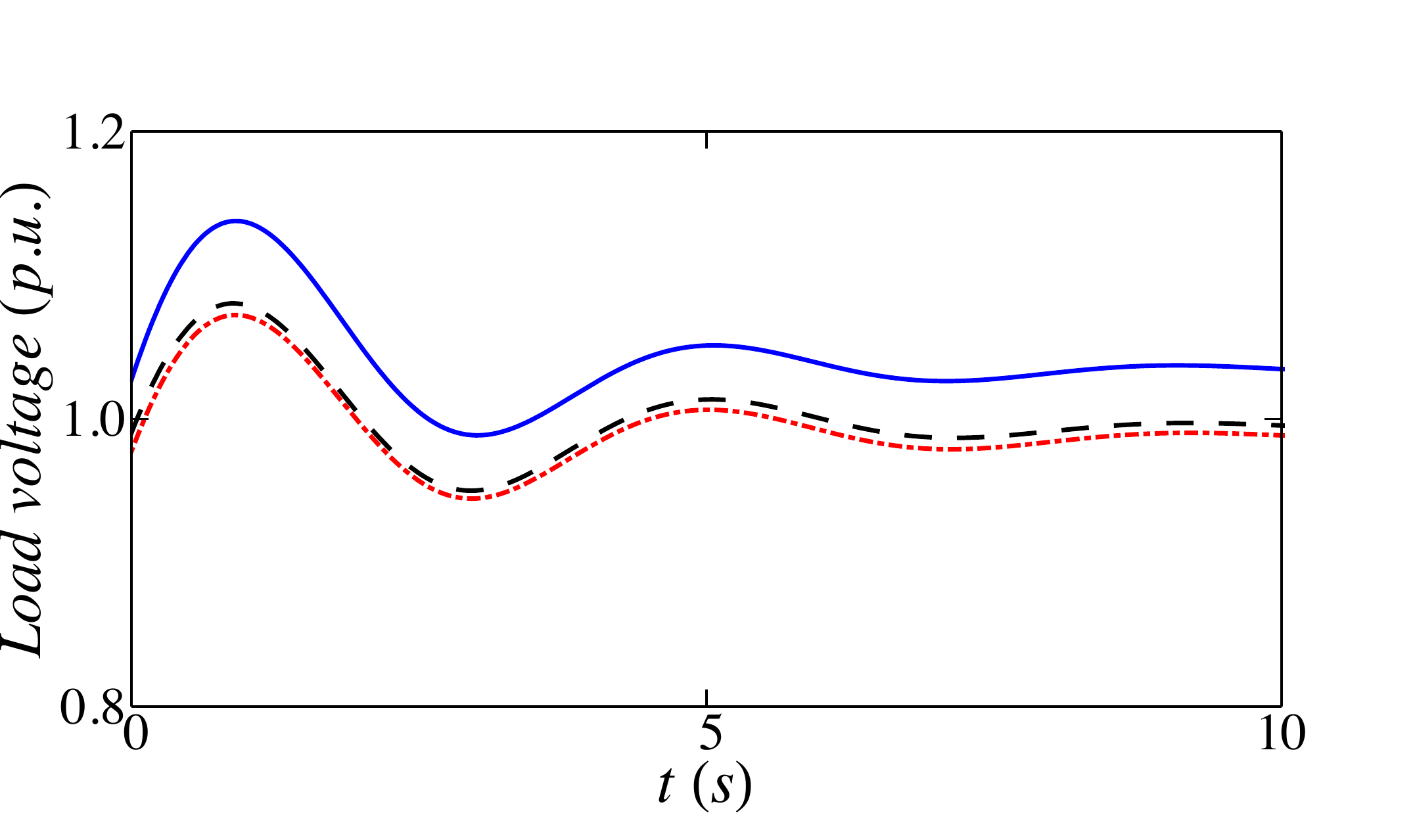}}}%
    \subfigure[$\circledS$ Stable]{{\includegraphics[width=4.5cm]{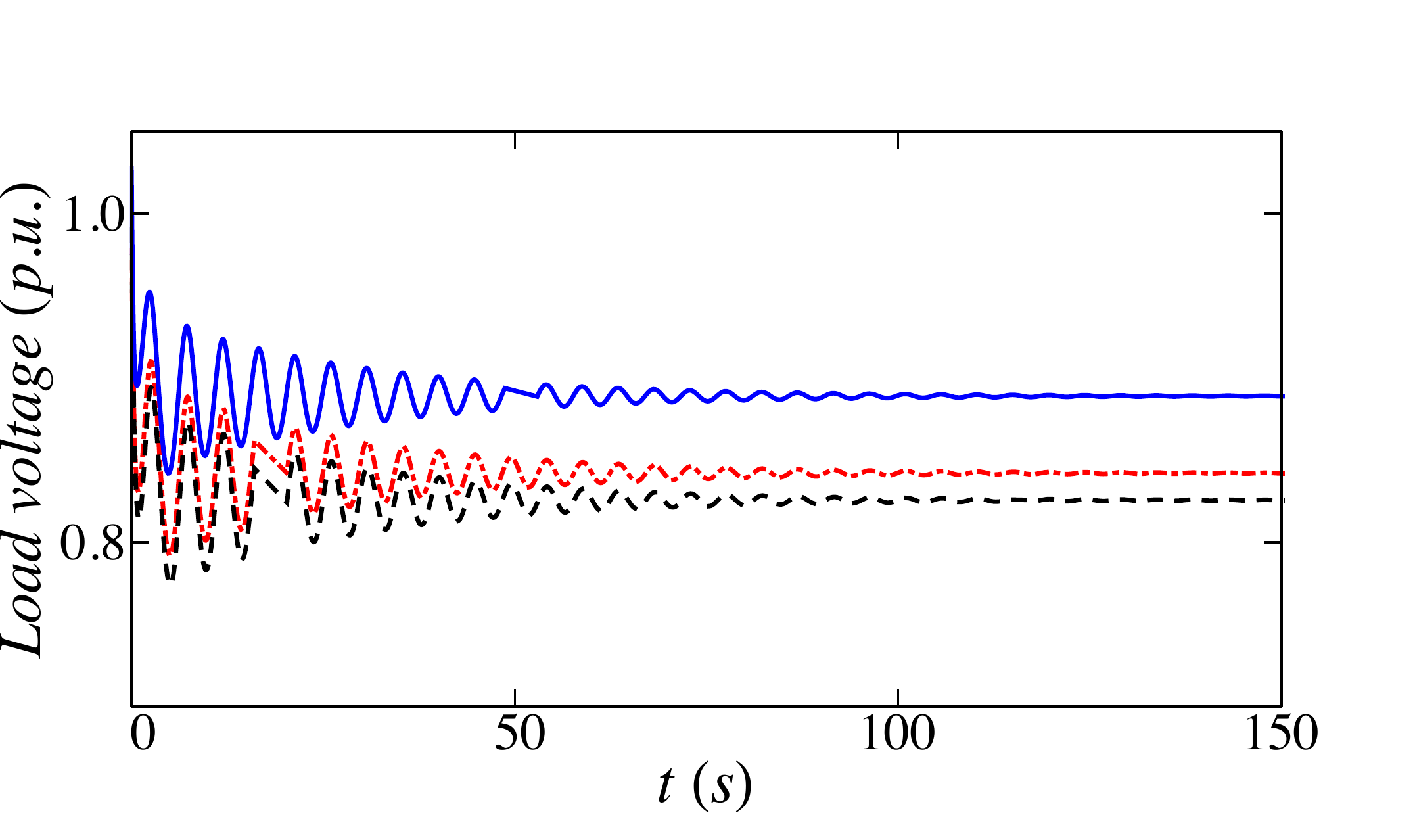} }}%
    \quad
    \subfigure[ {\ding[1.3]{173}} Limit Cycle]{{\includegraphics[width=4.5cm]{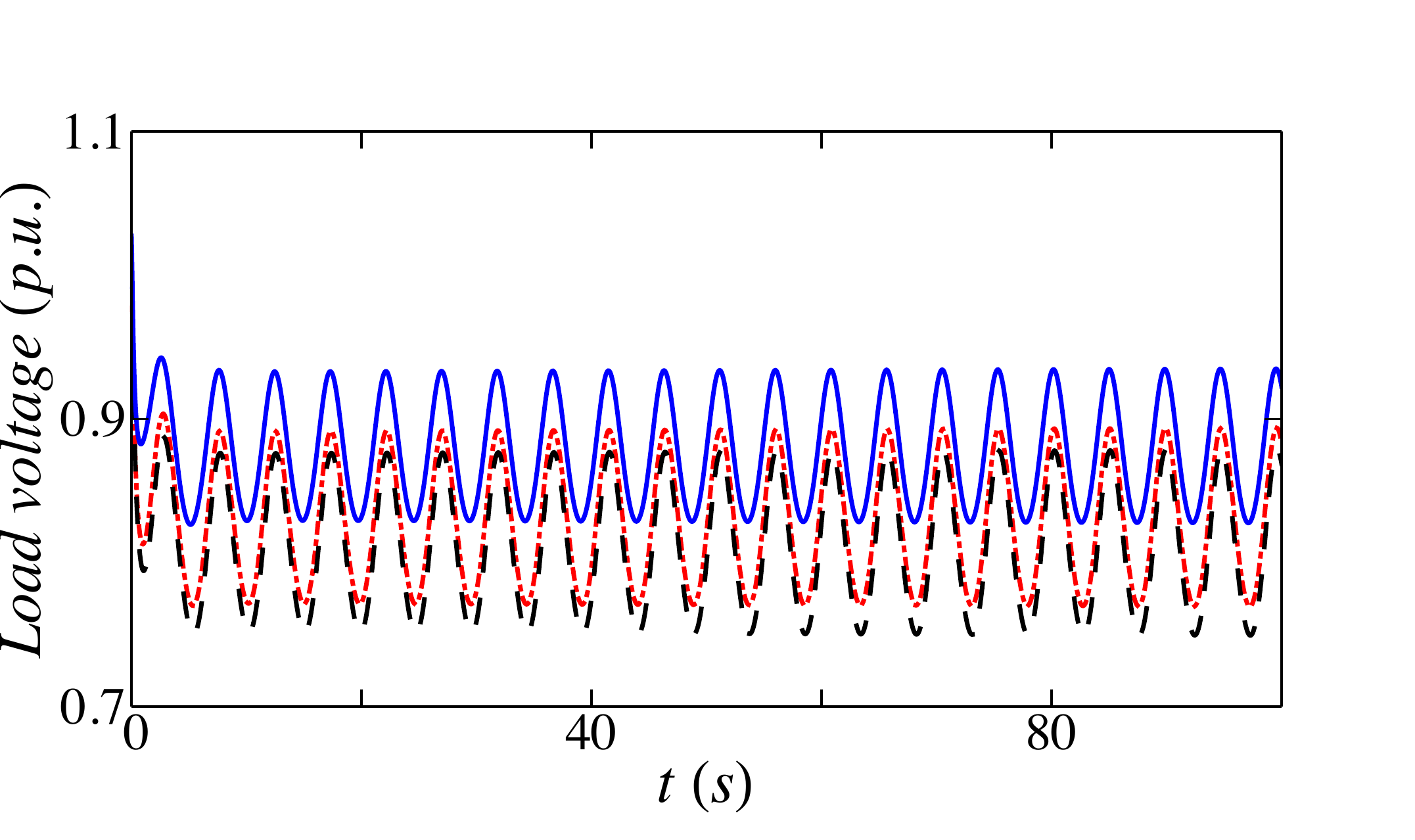}}}%
    \subfigure[ {\ding[1.3]{175}} Unstable (SNB)]{{\includegraphics[width=4.5cm]{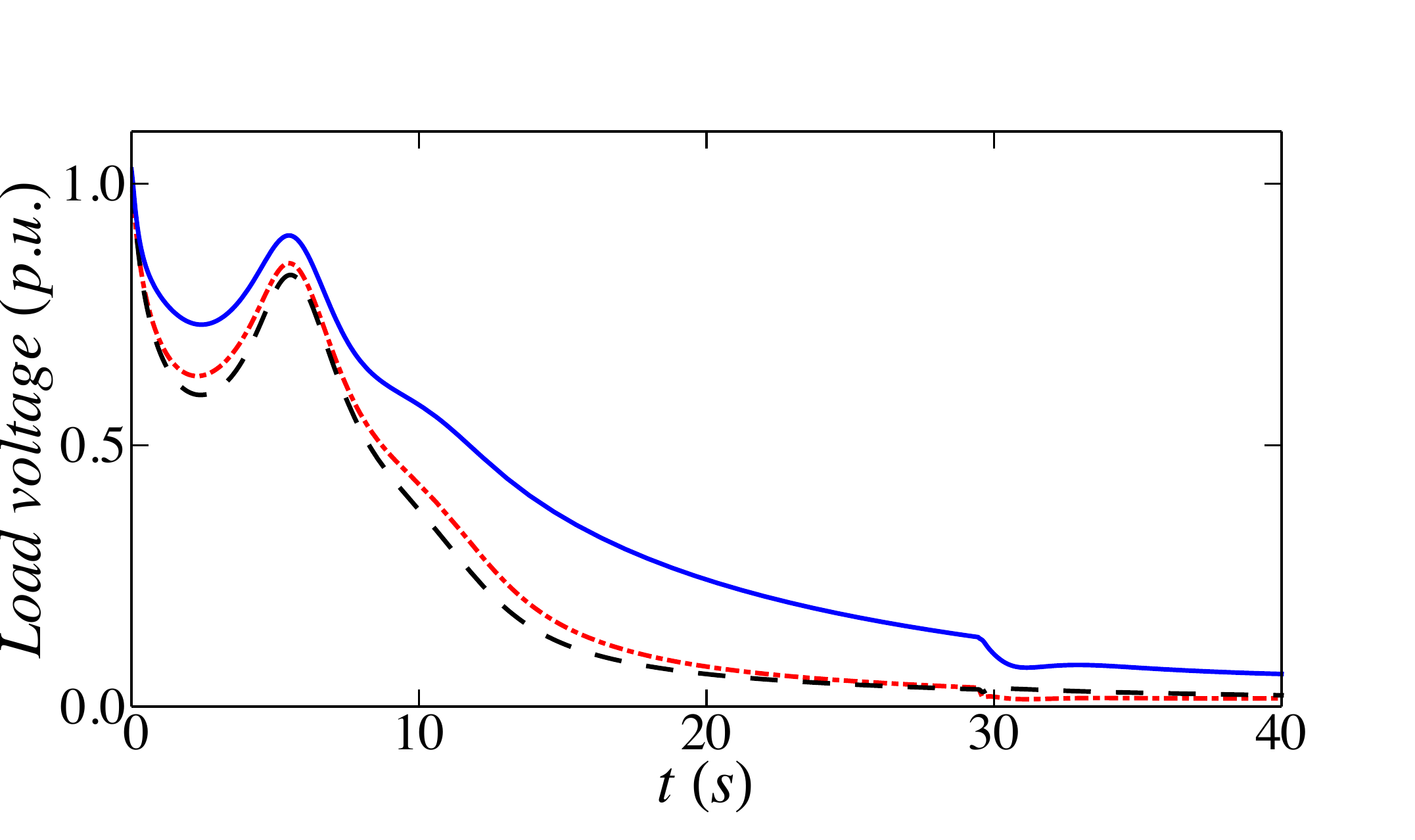} }}%
    \caption{The load voltage evolutions in time-domain simulations at different loading levels from {\ding[1.3]{172}} to {\ding[1.3]{175}} of the WSCC $3$-machine, $9$-bus system}%
    \label{fig:12s4}%
\end{figure}

\begin{figure}[htb]
    \centering
    \subfigure[$P_8=2\,p.u.$, Unstable]{{\includegraphics[width=4.5cm]{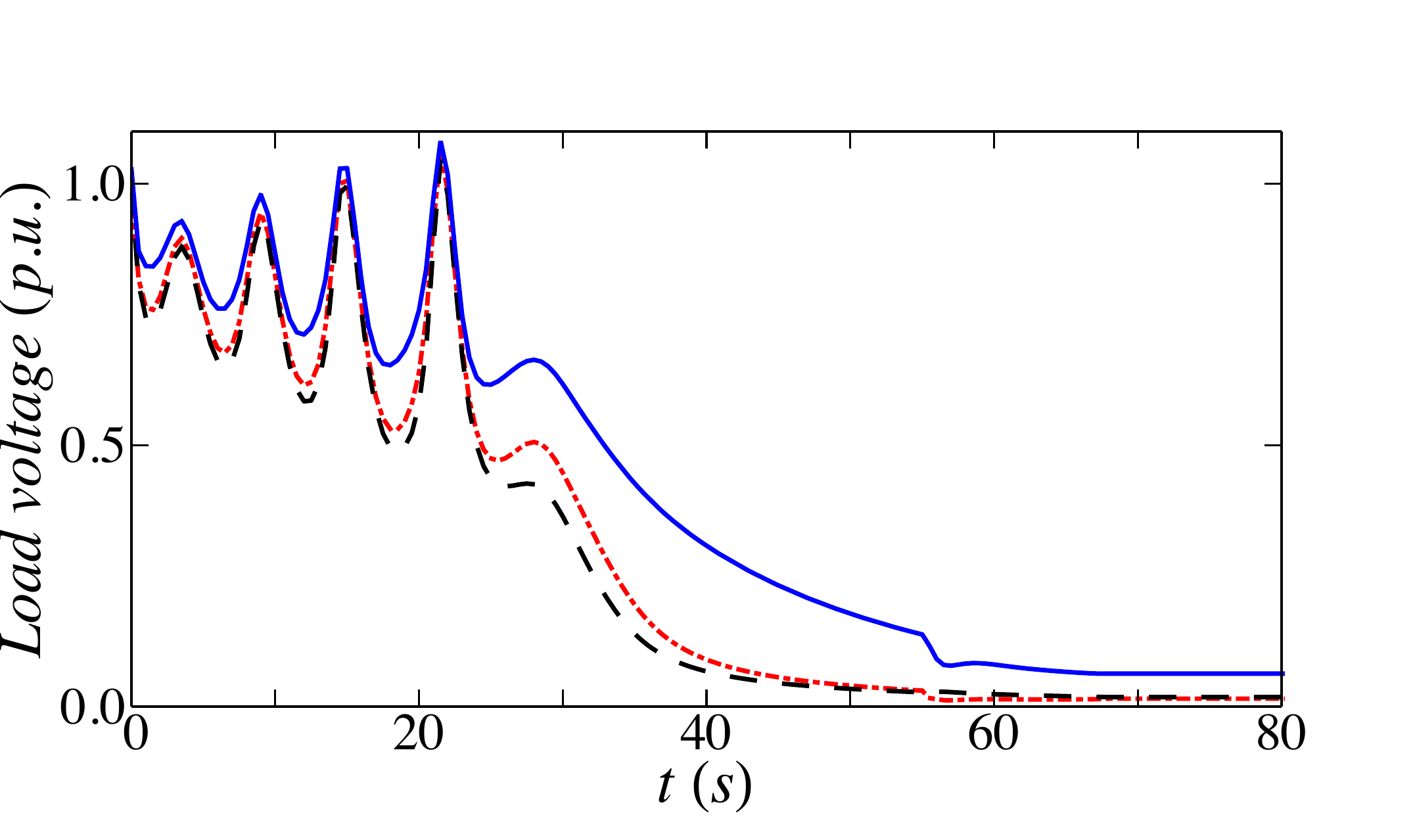}}}%
    \subfigure[Second Hopf bifurcation, Unstable]{{\includegraphics[width=4.5cm]{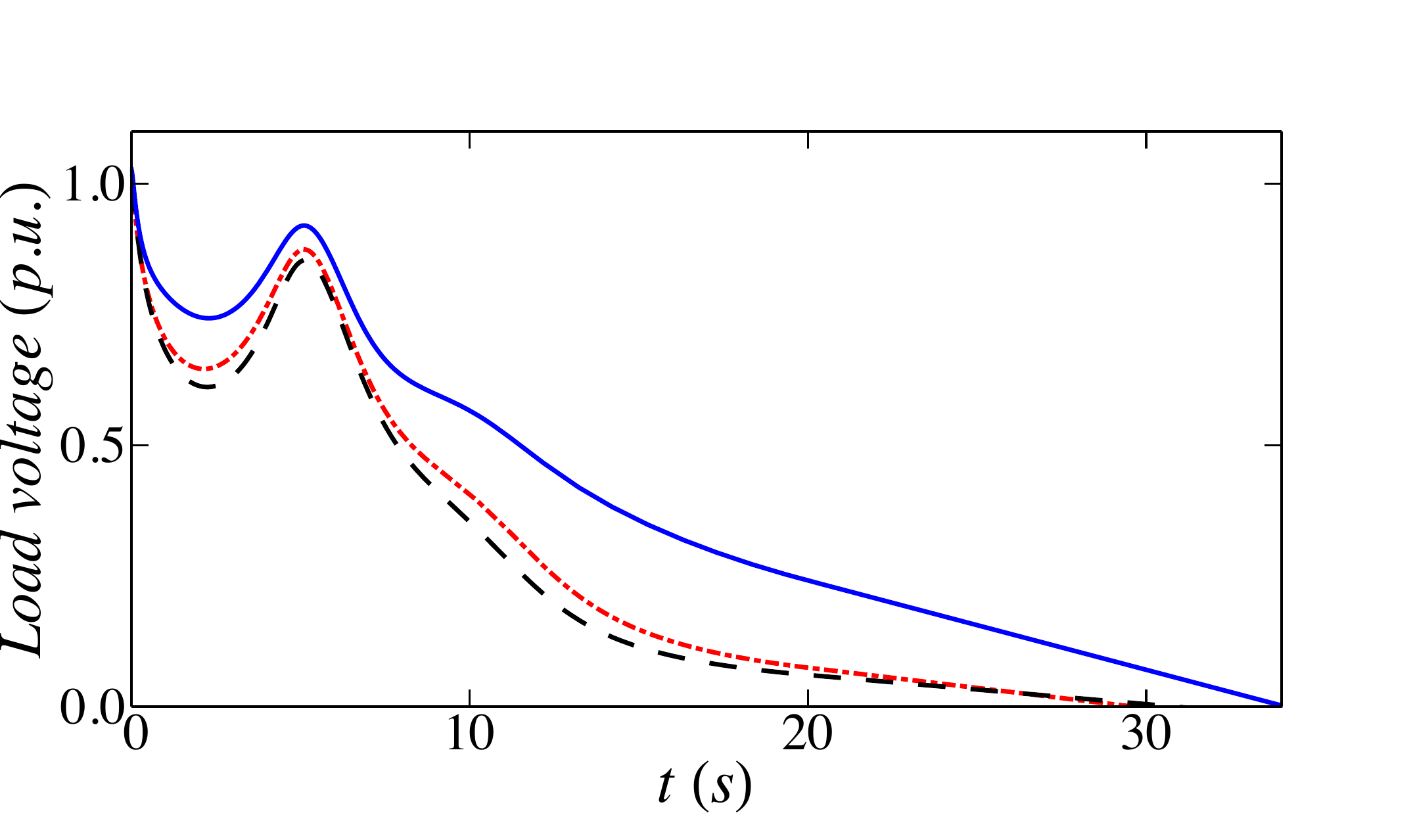} }}%
    \caption{The load voltage evolutions in time-domain simulations at $P_8=2\,p.u.$ and the second Hopf bifurcation of the WSCC $3$-machine, $9$-bus system}%
    \label{fig:24}%
\end{figure}

For the WSCC $3$-machine, $9$-bus system and the considered scenario with $\tau_5=6.5\,s$, $\tau_6=5.9\,s$, $\tau_8=5.35\,s$; the critical eigenvalue trajectories, \circled{1}-\circled{S}-\circled{2}-\circled{4} are plotted in Figure \ref{fig:crieigentraj_9bus}. In this case, as the load level increases from zero to the maximum loading level, the critical eigenvalue trajectory starts at \circled{1} or the point at ($-5.7$, $0$) which is far to the left, then follows the arrows direction to the origin or \circled{4}. The critical complex eigenvalue pair also crosses the imaginary axis to the right half plane then returns to the left half plane without coalescency. Along the trajectory the system encounters Hopf bifurcation twice. At \circled{4} where the critical real eigenvalue reaches the origin, SNB happens. Interestingly, there is a small region between the second Hopf bifurcation and SNB, the system is stable. However, in that region, low damping causes the system oscillates under the effect of a disturbance. The corresponding time-domain simulation also indicates that the initial condition need to close to the equilibirum state values to ensure that the system will converge to that equilibirum. This implies that the equilibrium has a small stability region. The trajectories in Figure \ref{fig:crieigentraj_2bus_tau735} and Figure \ref{fig:crieigentraj_9bus} are the two typical transients from Hopf bifurcation to SNB that can be observed when scaling the loading level. They may be different in the region between \circled{S} and \circled{4}, but in the end, one single real eigenvalue reaches the origin at \circled{4}.

The time-domain trajectories of the load voltages for corresponding power levels along the trajectory \circled{1}-\circled{S}-\circled{2}-\circled{4} are shown in Figure \ref{fig:12s4} where we use the same color code for the load voltages as in section \ref{sec:VSAexample}. Figure \ref{fig:24}(a) shows the load voltage levels at the load level between \circled{2} and the second Hopf bifurcation point, i.e. $P_8=2\,p.u.$. For $P_8=2.14\,p.u.$, the system encounters Hopf bifurcation again and the corresponding voltage trajectories at the loads are recorded in Figure \ref{fig:24}(b). In this case, the system loses stability via Hopf bifurcation. 

In the considered scenario, as the loading level increases beyond \circled{2} where Hopf bifurcation occurs, the stable limit cycle shrinks and disappears at $P_8=1.98\,p.u.$. Then if the loading level continues increasing, the system may collapse as shown in Figure \ref{fig:24}(a) or may converge to another stable equilibrium point if there is one. This is so because the eigenvalue analysis characterizes the stability of the linearized system corresponding to the considered equilibrium, but multiple stable equilibria can coexist at the same time. However, the latter case in which another stable equilibrium coexists is rather rare in the real power systems so the collapse scenario is more likely to happen. In general, in the non-certificated robust stability region between S and SNB, the system may exhibit different types of bifurcation such as Hopf bifurcation, transcritial bifurcation, and SNB \cite{ajjarapu1992bifurcation, canizares1994transcritical}.

\subsubsection{Effect of load power factors}

Qualitatively, the load power factor does not change the trajectory of the critical eigenvalues of the system within S-SNB. It mostly pushes the point on the real axis where the critical complex eigenvalues pair merge to the right and widens the distance between the two points on the imaginary axis at \circled{2}. The effect on \circled{3} is recorded in Table \ref{table:pfon3} for $\tau=7.35\,s$.

\begin{table}[ht]
    \caption{Effect of power factor on the critical eigenvalues}
    \centering
    \label{table:pfon3}
\begin{tabular}{|l||*{5}{c|}}\hline
\textbf{power factor}
&\makebox[3em]{$0.89\, lag$}&\makebox[3em]{$0.98\, lag$}&\makebox[3em]{$1.0$}&\makebox[3em]{$0.98\, lead$}&\makebox[3em]{$0.89\, lead$}
\\\hline\hline
\textbf{Re(s) @ \circled{3}} &0.65&1.10&1.20&2.31&3.56\\\hline
\end{tabular}
\end{table}

\subsubsection{Effect of the time constants of the loads}
\begin{figure}[!ht]
    \centering
    \includegraphics[width=1 \columnwidth]{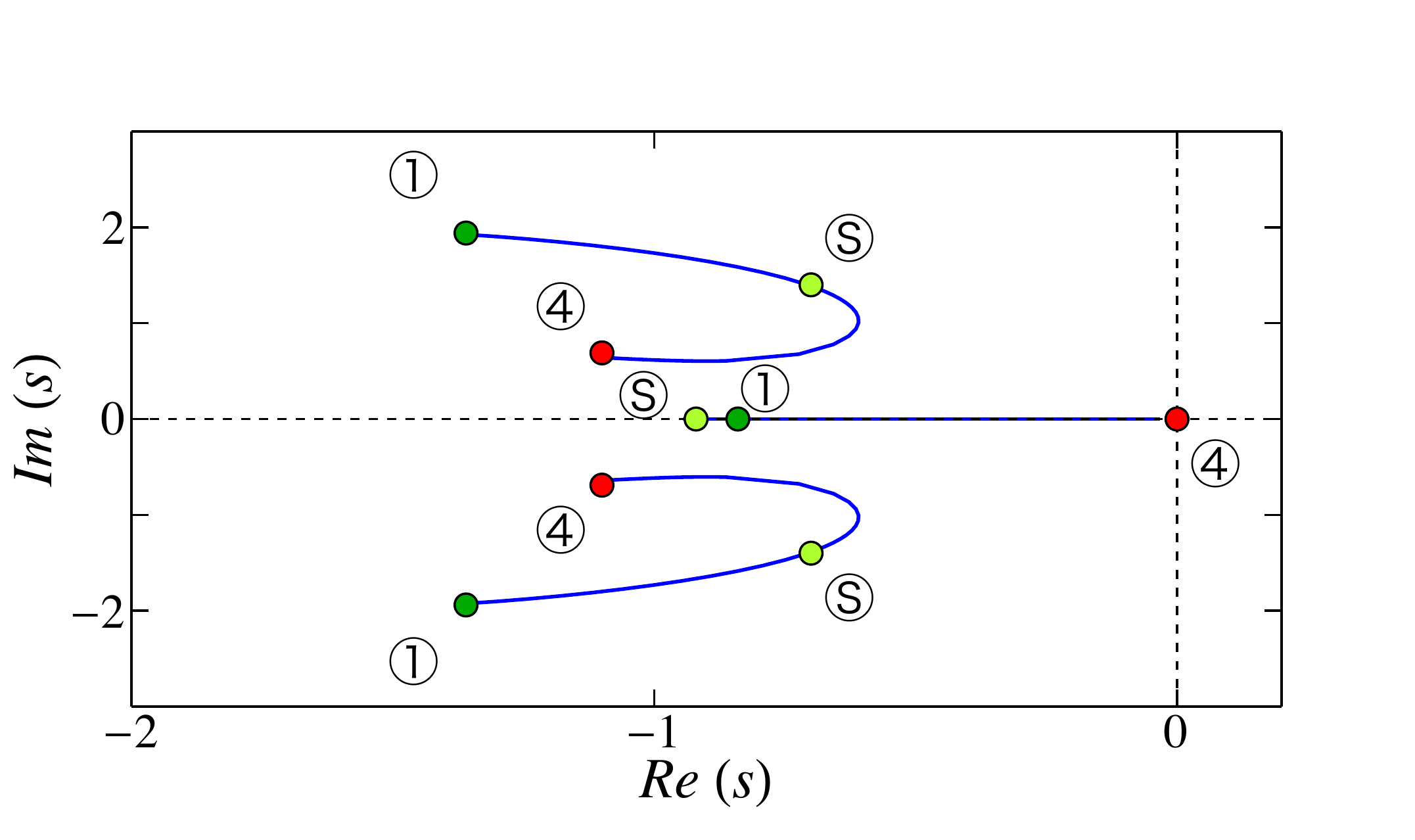}
	\caption{Critical eigenvalue trajectory under the load changes in the rudimentary system, $\tau=1\,s$}
     \label{fig:crieigentraj_2bus_tau1}
\end{figure}

For $\tau=1\,s$, the trajectory \circled{1}-\circled{S}-\circled{4} of critical eigenvalues of the system is plotted in Figure \ref{fig:crieigentraj_2bus_tau1}. In this case, Hopf bifurcation will not happen while increasing the loading level $P_0$, and all eigenvalues lie in the left half plane of the s-plane. At \circled{4}, the system encounters SNB or static voltage collapse. Moreover, the whole upper branch of the nose curve $PV$ is stable up to SNB.

When the instant relaxation time of the load increases to a large enough value, for example $\tau>7.35\,s$, the trajectory of the critical eigenvalues is similar to that in Figure \ref{fig:crieigentraj_2bus_tau735} except point \circled{3} on the real axis moves to the right. At the same time, \circled{S} also moves towards \circled{2} on the imaginary axis but it never reaches \circled{2}. This phenomenon can be explained as when the load time constant increases, the system may become unstable right after the robust stable point S. In this sense, if RSA cannot certify the system robust stability, the system is indeed non-robust stable. 

From our simulations we found that, if other parameters of the system are kept unchanged, the system is prone to be unstable if the instant relaxation times of the loads increase. This phenomenon can be understood as the larger time constants of the loads add more delay to the system which in turn reduces the phase margin \cite{franklin2010feedback}, finally causes the system to be unstable. In the s-plane, one can see that increasing the loads time constants pushes the critical eigenvalues to move close to the imaginary axis. When the critical eigenvalues cross to the right-half plane, the system is likely unstable.

\section{Conclusions and Future work}
In this work we have addressed the problem of uncertainty of load dynamics and its effect on the stability of the system and in particular on the occurrence of Hopf bifurcation. RSA developed in this work allows to certify the stability of the power system without making any assumptions on the dynamic response of the load. Whenever the system is certified to be robust stable, the system is guaranteed to be stable for any dynamic responses of the loads involved. The algorithm relies on convex optimization and can be applied even to large-scale system models. The regions that are certified to be robust stable are surprisingly large for models considered in the manuscript which suggest that Robust-Stability regime can be enforced in planning and operation without compromising efficiency and other economic factors. 

There are several ways of extending the algorithm that we plan to explore in future works. First, we plan to extend the types of uncertainties that can be handled to uncertainty in static characteristic, load levels, and allow for using the range bounds on the time constants. Second, we plan to develop algorithms that certify the robust stability of whole regions in parameter space, eliminating the need for repeating the procedure for every operating point candidate. Finally, we are interested in applying the algorithm to practical problems like stability constraint remedial action design, stability constraint planning and others. 

\section{Acknowledgement}
The work was partially supported by NSF, MIT/Skoltech and Masdar initiatives, Vietnam Educational Foundation, and the Ministry of Education and Science of Russian Federation, Grant Agreement no. 14.615.21.0001. We also thank Dr. Long Vu and Dr. Xiaozhe Wang for useful comments.

\bibliographystyle{IEEEtran}

\appendices
\section{The generic dynamic load model} \label{app:LM}
In this Appendix, we reproduce ULTCs and heating load models presented in \cite{Hill93} using the proposed generic dynamic load model. It is effective and convenient to represent the considered loads in the general form of (\ref{eq:abstract}). To illustrate this, we only present the models for active powers.

\subsection{ULTC dynamics}
We consider the ULTC depicted in Figure \ref{fig:ULTC}.
\begin{figure}[ht]
    \centering
    \includegraphics[width=0.6 \columnwidth]{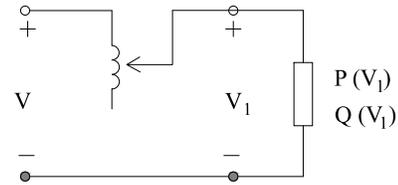}
	\caption{Tap-Changer and Static Load Combination \cite{Hill93}}
    \label{fig:ULTC}
\end{figure}

ULTC characteristics is adopted from \cite{Hill93} as follows:

\begin{equation} \label{eq:ULTCVapend}
\centering
V_1=K\,V
\end{equation}
\begin{equation}\label{eq:ULTCKapend}
T\dot{K}=-(V_1-V^0)
\end{equation}
\begin{equation}\label{eq:ULTCPd}
\centering
P=g\,V_1^2
\end{equation}
where $g=constant$; $V^0$ is the voltage set-point, for example $V^0=1\,p.u.$; $T$ represents the speed of tap changing; $K$ is the transformer ratio and $P$ is the power consumption level. 
Combine \eqref{eq:ULTCPd} with \eqref{eq:ULTCVapend}, yields:
\begin{equation} \label{}
\centering
P=gV_1^2=gK^2V^2=g_{eq}\,V^2
\end{equation}
where $g_{eq}$ is the equivalent conductance, $g_{eq}=gK^2$, then:
\begin{equation} \label{eq:UTLCdotgeq}
\centering
\dot{g_{eq}}=2gK\dot{K}
\end{equation}
From \eqref{eq:ULTCKapend} and \eqref{eq:UTLCdotgeq}, we have:

\begin{equation}
\label{eq:ULTCdoteq1}
\centering
\dot{g_{eq}}=-\frac{2gK}{T}(V_1-V^0)
\end{equation}
Since $g_{eq}=gK^2$, \eqref{eq:ULTCdoteq1} can be rewritten to  yields (\ref{eq:abstract}):

\begin{equation}
\label{eq:ULTCdoteq2}
\centering
\dot{g_{eq}}=-\frac{2}{T}\sqrt{g\,g_{eq}}(V_1-\,V^0)
\end{equation}

\subsection{Heating load dynamics}

Consider the heating load model in Figure \ref{fig:heatingmodel} \cite{Hill93}.
\begin{figure}[ht]
    \centering
    \includegraphics[width=0.4 \columnwidth]{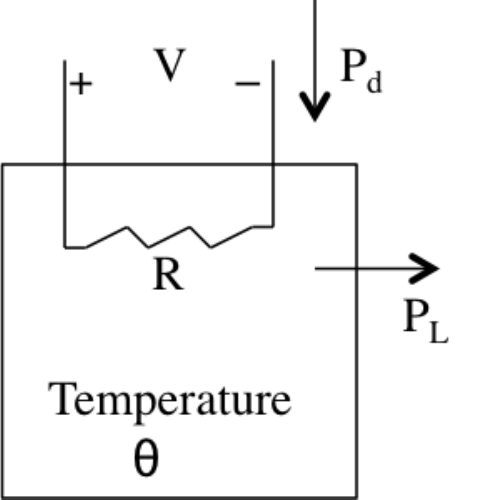}
	\caption{Room model \cite{Hill93}}
    \label{fig:heatingmodel}
\end{figure}
The heating load's characteristics are given as: 

\begin{equation}
\label{eq:heatT}
\centering
T\dot{\theta}=P_d-P_L
\end{equation}
where the power demand $P_d=\frac{V^2}{R(\theta)}$, $P_L$ is the losses. The load conductance can be computed as:
\begin{equation} \label{eq:heatgtheta}
\centering
g=\frac{P_d}{V^2}=\frac{1}{R(\theta)}=f(\theta)
\end{equation}
Differentiating the two sides of \eqref{eq:heatgtheta}, yields:

\begin{equation}
\label{eq:heatdotg}
\centering
\dot{g}=\frac{df}{d\theta}\dot{\theta}
\end{equation}

Substituting \eqref{eq:heatT} into \eqref{eq:heatdotg}, we have:

\begin{equation} \label{eq:heatdotg1}
\centering
\dot{g}=\frac{df}{d\theta}\frac{P_d-P_L}{T}
\end{equation}
If linear resistance characteristic is applied, i.e. $R(\theta)=r\theta$, (\ref{eq:heatdotg1}) becomces:

\begin{equation}\label{eq:heatdotg2}
\dot{g}=-\frac{1}{Tr\theta^2}(g\,V^2-P_L)
\end{equation}
Since $\theta=\frac{1}{rg}$, \eqref{eq:heatdotg2} represents the proposed generic dynamic load model \eqref{eq:abstract}.

\section{The WSCC $3$-machine $9$-bus system data} \label{app:9busdata}

\begin{table}[ht]
    \caption{The WSCC $3$-machine $9$-bus system branch data \cite{PaiDynamics}}
    \centering
    \label{9busbranchdata}
    \begin{tabular}{|c |c |c|c |}
    \hline
    \textbf{From} & \textbf{To} & \textbf{Line impedance}& \textbf{Half shunt capacitance}  \\
    \textbf{bus} & \textbf{bus } & \textbf{$(p.u.)$}& \textbf{$(p.u.)$}\\
               \hline
    $4$ & $5$ & $0.01+j0.085$&$0.088$  \\
    $4$ & $6$ & $0.017+j0.092$ &$0.079$   \\
    $5$ & $7$ & $0.032+j0.161$ &$0.153$ \\
    $6$ & $9$ & $0.039+j0.17$ &$0.179$ \\
    $7$ & $8$ & $0.0085+j0.072$ &$0.0745$ \\
    $8$ & $9$ & $0.0119+j1.008$ &$0.0145$ \\
    
    \hline
    \end{tabular}
\end{table}

\begin{table}[ht]
    \caption{The WSCC $3$-machine $9$-bus system transformer data \cite{PaiDynamics}}
    \centering
    \label{9bustransdata}
    \begin{tabular}{|c |c |c|c |}
    \hline
    \textbf{From bus} & \textbf{To bus} & \textbf{Impedance $(p.u.)$}& \textbf{Tab}  \\
               \hline
    $1$ & $4$ & $j0.0576$&$16.5/230$  \\
    $2$ & $7$ & $j0.0625$ &$18.0/230$   \\
    $3$ & $9$ & $j0.0586$ &$13.8/230$ \\
    \hline
    \end{tabular}
\end{table}

\begin{IEEEbiography}[{\includegraphics[width=1in,height=1.25in,clip,keepaspectratio]{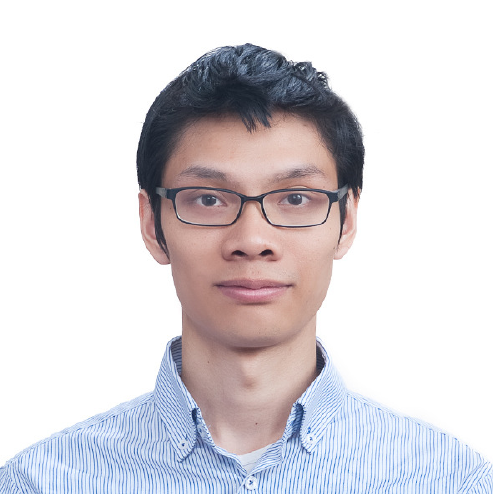}}]{Hung D. Nguyen} (S`12) was born in Vietnam, in 1986. He received the B.E. degree in electrical engineering from Hanoi University of Technology, Vietnam, in 2009, and the M.S. degree in electrical engineering from Seoul National University, Korea, in 2013. He is pursuing a Ph.D. degree in the Department of Mechanical Engineering at Massachusetts Institute of Technology (MIT). His current research interests include power system operation and control; the nonlinearity, dynamics and stability of large scale power systems; DSA/EMS and smart grids.
\end{IEEEbiography}
\vspace{-160 mm}
\begin{IEEEbiography}[{\includegraphics[width=1in,height=1.25in,clip,keepaspectratio]{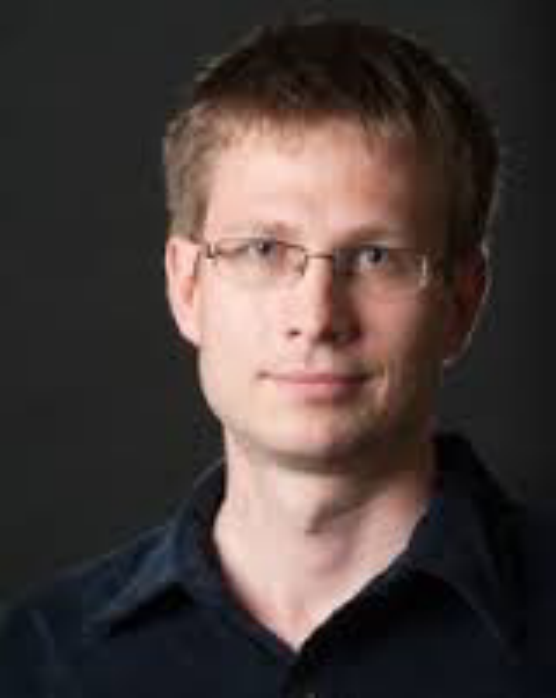}}]{Konstantin Turitsyn} (M`09) received the M.Sc. degree in physics from Moscow Institute of Physics and Technology and the Ph.D. degree in physics from Landau Institute for Theoretical Physics, Moscow, in 2007.  Currently, he is an Assistant Professor at the Mechanical Engineering Department of Massachusetts Institute of Technology (MIT), Cambridge. Before joining MIT, he held the position of Oppenheimer fellow at Los Alamos National Laboratory, and Kadanoff–Rice Postdoctoral Scholar at University of Chicago. His research interests encompass a broad range of problems involving nonlinear and stochastic dynamics of complex systems. Specific interests in energy related fields include stability and security assessment, integration of distributed and renewable generation.
\end{IEEEbiography}

\end{document}